\documentclass[10pt]{article}

\usepackage{setspace}
\usepackage{epstopdf}
\usepackage{graphicx}
\usepackage{url}
\usepackage{amsmath,amssymb,amsfonts}
\usepackage{amsthm}
\usepackage{subfigure}
\usepackage{wrapfig}
\usepackage{fancyhdr}
\usepackage{authblk}
\usepackage{geometry}
{
\cfoot{Approved for public release: Distribution unlimited}
}

\renewcommand{\thispagestyle}[1]{} % do nothing; this is to make footer come on first page

\newtheorem{theorem}{Theorem}[section]
\newtheorem{lemma}[theorem]{Lemma}

\newtheorem{Corollary}[theorem]{Corollary}

\title{Census: Fast, scalable and robust data aggregation in MANETs
}

\author[1]{V. Kulathumani\thanks{vinod.kulathumani@mail.wvu.edu}}
\author[2]{A. Arora\thanks{anish.arora@samraksh.com}}
\author[2]{M. Sridharan\thanks{mukundan.sridharan@samraksh.com}}
\author[2]{K. Parker\thanks{kenneth.parker@samraksh.com}}
\author[1]{M. Nakagawa\thanks{manakagawa@mix.wvu.edu}}
\affil[1]{Department of Computer Science, West Virginia University, Morgantown WV 26505}
\affil[2]{The Samraksh Company, Dublin OH 43017}

\begin{document}

\thispagestyle{fancy}
\maketitle

\doublespacing

\begin{abstract}
This paper describes {\em Census}, a protocol for data aggregation and statistical counting in MANETs. {\em Census} operates by circulating a set of tokens in the network using {\em biased random walks} such that each node is visited by at least one token. The protocol is structure-free so as to avoid high messaging overhead for maintaining structure in the presence of node mobility. It biases the random walks of tokens so as to achieve fast cover time; the bias involves short albeit multi-hop gradients that guide the tokens towards hitherto unvisited nodes. {\em Census} thus achieves a cover time of $O(N/k)$ and message overhead of $O(N log(N)/k)$ where $N$ is the number of nodes and $k$ the number of tokens in the network. Notably, it enjoys scalability and robustness, which we demonstrate via simulations in networks ranging from $100$ to $4000$ nodes under different network densities and mobility models.
\end{abstract}

\let\thefootnote\relax\footnotetext{The email point of contact is vinod.kulathumani@mail.wvu.edu. This work was supported in part by Defense Advanced Research Projects Agency (DARPA) contract FA8750-12-C-0278. The views, opinions, and/or findings contained in this paper are those of the authors and should not be interpreted as representing the official views or policies of the Department of Defense or the U.S. Government.}


\section{Introduction}
This paper presents {\em Census}, a fast, scalable and robust protocol for data aggregation in MANETs. {\em Census} is a structure-free protocol that relies on biased random walks to achieve aggregation. The protocol operates by circulating a set of one or more tokens using biased random walks, in such a way that every node in the network is visited by at least one token. We say that a node is {\em visited} by a token when the node gets exclusive access to the token; the visitation period can be used by the node to add node-specific information into the token, resulting in data aggregation.  Note that the concept of visiting all nodes individually differs from that of token dissemination \cite{mmgossip, trickle} over the entire network where it suffices for every node to simply hear at least one token, as opposed to getting exclusive access to a token. 

There are many applications for an all node visitation service, such as voting, computing aggregates (i.e., max, min, or average), statistical counting (i.e., estimating the fraction of nodes that satisfy a state predicate), and computing empirical distributions of data in a network \cite{learnhist}. Example application scenarios include computing average sensor measurements, counting battalions and ammunition in military networks, and computing aggregate traffic densities in vehicular networks. 

%The node-visitation phase by itself could also be used to provide every node an access to a critical resource in a mutually exclusive manner, such as the use of a shared high-bandwidth link. 

In static networks with stable links, data aggregation can be realized by traversing fixed routing structures such as trees or network backbones \cite{naik_sprinkler, tag, ctp, directdiff}. However, in mobile networks and networks with frequent link changes, topology driven structures are likely to be unstable and to incur a high communication overhead for maintenance.  In a recent paper, we have analyzed why routing protocols for MANETs such as OLSR \cite{rtsw3} are unable to scale beyond $100$-$150$ nodes, in terms of a {\em scaling wall} \cite{rtswarxiv}. Firstly, as network size increases, paths are more likely to fail --- the path connectivity interval falls as $O(\sqrt{N})$, with increasing network size $N$. Secondly, even small changes in node speed significantly increase this path failure rate and instability. Finally, in order to be successfully used for routing, the broken paths need to be fixed much faster than the rate at which they break and this involves fast link estimation for discovering broken paths and then repairing those paths, both of which incur high message overhead. 

%Finally, these protocols decrease in scalability as the average node speed increases. 

Thus, in contrast to the static network case, {\em Census}, exploits the simplicity of random walks to achieve token coverage in MANETs. Random walks are attractive for MANETs because they are inherently stable in the presence of network dynamics, have no critical points of failure, avoid structure maintenance, and have very little state overhead. In fact, {\em Census} is largely mobility agnostic --- motion models and node speeds have hardly any impact on its performance, as a major component of the protocol is simply token passing. 
%in fact, we observe that high node speeds seem to benefit convergence time by enabling faster mixing of nodes. 
In a pure random walk, a node that holds a token picks a random node in its neighborhood and transfers the token to that node. The first time that any token visits a node can be used by the node to add its state into the aggregate being computed. This process is repeated until all nodes have been visited. Unfortunately, the {\em cover time} for random walks (time to visit all nodes) is typically high because of wasted exploration when a token repeatedly encounters already visited nodes. In order to expedite the cover time, we explore in this paper the idea of partially guiding random walks towards unvisited nodes.

%n order to expedite the cover time using random walks, we locally bias the random walk towards unvisited nodes in its one hop neighborhood. Whenever a node is visited for the first time by one of the tokens, the node adds its state into the aggregate being computed and passes on the token. 

{\bf Local bias:}~~Let us first consider biasing based on local information only, by giving preference to unvisited nodes in the neighborhood when forwarding the token.  So if there are one or more unvisited nodes in the direct one-hop neighborhood (i.e., within the communication range) of a token, the token is passed to one of the unvisited nodes chosen at random. We show analytically that local bias, by itself, achieves a cover time of $O(Nlog(N)/k)$ where $k$ is the number of tokens being used. We moreover prove that by just using local bias, a significant portion of the network can be covered without significant wasted exploration at already visited nodes. However, when the fraction of already visited nodes in the network rises beyond a certain threshold, local biasing exhibits a slowdown. This is because when all the nodes within the communication range of a token holder are already visited, the scheme reduces to a canonical random walk until an unvisited node becomes a neighbor. While the order of convergence in relation to $N$ remains $O(Nlog(N))$, the slowdown creates a long tail in the convergence and significantly increases cover time.  To redress this shortcoming, we use a complementary method that further speeds up the cover time, which forms the basis of {\em Census}.

%{\bf Multi-hop gradient bias:}~~To prevent random walks from getting stuck in regions of visited nodes while there are still unvisited nodes to be explored, we set up short, temporary multi-hop gradients to pull the token towards unvisited nodes. For Census with gradient bias, we show a cover time of $O(Nlog(N)/kd)$, where $d$ is the average network density and $k$ is the number of tokens being used. While the convergence order in terms of network size is same as that of local biasing, the gradient bias avoids the long tail and reduces both cover time and the token passing overhead. The gradient assisted random walk introduces a gradient message overhead of $O(Nlog(N)/k)$, to pull tokens towards unvisited nodes. However, this message overhead is compensated by a reduction in the required number of token transfers. In fact, we show that the ratio of token transfers to that of node size remains under $3$ for networks of upto $4000$ nodes.

{\bf Multi-hop gradient bias (Census):}~~To prevent the random walks from getting stuck in regions of visited nodes while there are still unvisited nodes to be explored, we set up short, temporary multi-hop gradients to pull the token towards unvisited nodes. We show analytically that this yields a cover time of $O(N/k)$, where $k$ is the number of tokens being used. Thus, the order of convergence improves by a factor of $log(N)$. In doing so, {\em Census} introduces a gradient message overhead of $O(Nlog(N)/k)$, to pull tokens towards unvisited nodes. Nevertheless, this overhead is compensated by a reduction in the required number of token transfers. In fact, our simulations show that the ratio of token transfers to that of node size, i.e., the exploration overhead of gradient biased random walk remains close to $2$ even for networks as large as $4000$ nodes.

\vspace*{-1mm}
\subsection{Summary of contributions}

We introduce the idea of applying biased random walks for the node visitation problem in MANETs, which to the best of our knowledge has not been explored before. We show analytically that {\em Census} has a fast cover time of $O(N/k)$ in terms of network size $N$ and token count $k$. We also show that the overall communication cost for {\em Census} is $O(Nlog(N)/k)$. We characterize the impact of network density on the cover time and message overhead. We analytically compare the performance of {\em Census} with regular random walks, flooding, gossip, diffusion and structure based routing protocols as well. We corroborate all of our analytical results for {\em Census} using ns-3 based simulations of mobile networks ranging from $100$ to $4000$ nodes, under different network densities and mobility models. Our simulation results demonstrate that biased random walks are simple, yet effective tools for achieving scalable and robust token coverage in MANETs. 

\vspace*{-1mm}
\subsection{ Outline of the paper}

In Section $2$, we describe the system model. In Section $3$, we describe the {\em Census} protocol. In Section $4$, we present an analytical characterization of convergence time and message overhead for {\em Census}, and compare this against that of random walks with local bias and pure random walks (without bias). In Section $5$, we evaluate the performance of Census and compare it against pure random walks, random walks with local bias, flooding, gossip, diffusion and structure-based aggregation protocols. In Section $6$, we discuss implementation considerations for {\em Census} in a MANET such as handling message losses and termination detection. We discuss related work in Section $7$ and make concluding remarks in Section $8$.

%\vspace*{-1mm}
\section{System model}
We consider a mobile network of $N$ nodes deployed over a two dimensional region whose width and height scale as $\theta(\sqrt{N} \times \sqrt{N})$ and whose communication range and density are constant irrespective of network size $N$. For our analysis, we focus on a random walk mobility model \cite{rw10, rw11} for the nodes, whereas for purposes of simulation we consider a variety of mobility models such as random waypoint and Gauss-Markov. In the random walk mobility model, at each interval a node picks a random direction uniformly in the range $[0,2\pi]$ and moves with a constant speed that is randomly chosen in the range $[v_l,v_h]$ for a constant distance $\gamma$. At the end of each interval, a new direction and speed are calculated. This model is Brownian in its characteristics; the Brownian model can be described as a scaling limit of this motion model under small step sizes \cite{rw9}. The random walk motion model results in node locations that are uniformly distributed across the network \cite{rw10}. Therefore, although the density of the nodes is time varying, we note that over time the average number of nodes per unit disk communication range is constant (denoted as $d$).

One or more tokens are introduced at random locations that are uniformly distributed within the network. The objective of {\em Census} is to pass the tokens around the network such that every node in the network is visited by at least one token. Recall that we say that a node is visited by a token when it gets exclusive access to the token.

\section{Census protocol}
{\em Census} consists of two components: (i) token passing, and (ii) gradient setup. To realize these components, each node stores three variables, {\em visited}, {\em holder} and {\em level}. The variable {\em visited} is a Boolean value that tracks whether a node has been visited by any of the tokens; {\em visited} is initially {\em false} at all nodes. When a token first arrives at a node, {\em visited} is set to $true$. Tokens are assumed to be initiated at a random set of nodes. All nodes in which a token is initiated are marked as visited by default and the token value is initialized to the data at the corresponding node. The variable {\em holder} is used to identify nodes that currently hold a token. When a gradient bias is used, each node also participates in a gradient setup process to attract tokens towards unvisited nodes. To do so, each node uses the state variable level where $0 \leq level \leq 1$. Nodes that are unvisited are at $level=1$. Nodes that hold a token set $level$ to $0$ as soon as they receive a token.

\subsection{Token passing}

\subsubsection{Token passing with local bias only}

\begin{figure}[t]
  \begin{center}
    \includegraphics[width=.35\textwidth]{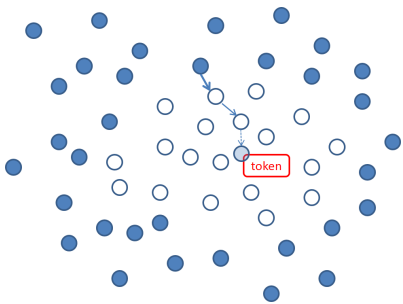}
    \caption{During token passing, a token may be surrounded by an island of visited nodes (white circles), i.e., all neighboring nodes have already been visited. Nodes that have not yet been visited (indicated by dark circles) periodically set up a gradient using the set of visited nodes to attract the token towards them.}
    \label{fig:gradient}
  \end{center}
\end{figure}

For random walks with only a local bias, a token holder announces that it has a token. Nodes that hear this message and have not been visited, start a timer to send a request at a random slot within a chosen interval $[0, .. , \frac{T_r}{2}]$. Nodes that hear this message and have already been visited, start a timer to send a request at a random slot within a chosen interval $[\frac{T_r}{2}, .. , T_r]$. This ensures that unvisited nodes get a chance to transmit before visited nodes. Also, if a node hears another request being sent, it suppresses its own request. This ensures that the number of requests being sent remain low and fairly constant irrespective of the network density. 

A timer $T_r$ is started at the token holder to accept requests for the token. The token holder picks a random unvisited node if at least one unvisited node sends a request. Otherwise, the token holder picks a random visited node. The token is transferred to the chosen node. The node that receives the token marks itself as visited if it was unvisited so far. If the token is used for data aggregation, an already visited node may not add its information again to a token. This concludes the procedure for token passing using random walks with only a local bias. The token is continued to pass iteratively using this procedure. 

\subsubsection{Token passing with gradient bias}

For random walks with a gradient bias (i.e. Census), a token holder announces that it has a token. Nodes that hear this message send a request at a random slot within a chosen interval $T_r$. A timer $T_r$ is started at the token holder to accept requests for the token. All nodes with $level>0$ randomize their response time and reply to the token announcement along with their current level. Nodes with $level>0$ are nodes that have either not been visited $(level=1)$ or nodes that have been visited and are now part of a gradient $(0 < level < 1)$. If a node hears another request being sent with a $level$ greater than itself, it suppresses its own request. The token holder stores all requests received during time $T_r$. The replies are sorted based on the level of the requestors and the token is sent to the node with the highest level. When multiple requestors exist with the same level, the token recipient is chosen randomly among that set. Thus if any unvisited node requests a token, the token will be sent to that node. If all nodes that have currently requested the token have been visited, the token is sent to the node with the highest value of level, which is expected to be the node that is closest to an unvisited node. As soon as a token reply has been sent, the node resets holder to $0$. The following section describes how short multi-hop gradients are setup to attract tokens towards unvisited nodes.

\subsection{Gradients}

\subsubsection{Gradient setup}

During Census operation a token can get stuck inside a region where all its neighbors have already been visited. To recover from such a scenario, a gradient is setup in the network to attract tokens towards unvisited nodes, i.e., nodes with $level=1$ (See Figure 1). This is done as follows. Nodes with $level=1$, for which none of their neighbors currently hold a token and have at least one neighboring node with $level=0$, initiate a gradient setup by broadcasting a gradient message. Nodes with $level=0$ that receive a gradient message from update their level to half of sender's level and rebroadcast the gradient message. Thus, gradient broadcasts propagate only till the region where nodes with non-zero level are present, filling up the gap between an unvisited node and other nodes with non-zero levels. 

\subsubsection{Gradient refresh}

To account for node mobility, gradients have to be periodically refreshed. To do so, when a node updates its level from zero to some non-zero value $< 1$, it starts a timer proportionate to the new level and when the timer expires it resets its level back to $0$. Thus nodes with higher values of level are refreshed slower than smaller values. This heuristic is based on two reasons. (1) Gradients should preferably not be refreshed before a token is able to climb up a gradient and reach an unvisited node. By refreshing at a rate proportional to the value of level, a token gets more time to reach closer to the source of the gradient. (2) Nodes that are far away from an unvisited node (closer to the bottom of the gradient) should prevent blocking of gradient setup from unvisited nodes that are nearby, for extended periods of time.

\section{Census analysis}
For mesh networks modeled as geometric graphs with uniform degree of connectivity, the expected cover time for pure random walks (no biasing) is known to be $O(N log^2 N)$ \cite{rw7}. However, pure random walks are likely to incur substantial wasted exploration, wherein a token repeatedly encounters nodes that are already visited. In this section, we characterize how biasing the tokens towards unvisited nodes improves this bound. We quantify the expected bounds on cover time as well as message overhead for both local and gradient biasing. We also characterize the impact of multiple tokens on the convergence time and message overhead. Our results show that random walks with local bias have a cover time of $O(Nlog(N)/k)$, and {\em Census} further reduces the cover time to $O(N/k)$ in terms of network size using short, multi-hop gradient bias. In the remainder of this section, we analytically compare Census with traditional techniques for data aggregation such as structure based protocols, flooding, and gossip.

%Apart from the order improvement of $log^2(N)$, 

To begin with, let us prove a lemma that characterizes the expected geometric distance between unvisited tokens over time. This is useful because in the node visitation problem as long as there is even a single unvisited neighbor, a token will be transferred to that node. To this end, we map our requirement to that of determining the minimum network density for connectivity in a wireless network and use nearest neighbor methods \cite{mindensity} that are often used in these analyses.  

\begin{lemma}
There exists at least one unvisited node within $h$-hops of a token holder in Census with probability $p$ as long as the fraction of unvisited nodes in the network $z$ satisfies $z>\theta/(h^2d)$, where $\theta = -ln(1-p)$. \label{lemma1}
\end{lemma}

\begin{proof}
Let $n_u$ denote the number of unvisited nodes at a given time in the network and let $\rho_u$ denote the uniform density of the network. Note that a token may be transferred to any unvisited neighbor with equal probability resulting in the {\em random walk traversing the network uniformly}. It follows that it is not the case that a token executing a random walk will visit all nodes that are clustered together before moving out. At each step, the choice of the next node to traverse is random. Therefore, given that a network is connected and uniformly distributed, a random walk of the tokens is expected to traverse the network uniformly. And it is indeed the case that a random walk 2d mobility model results in node locations that are uniformly distributed across the network \cite{rw10}. 

We now recall a result on the distribution of nearest neighbor distance in two dimensions. For a homogeneous Poisson point process in two dimensions, the probability density function of nearest neighbor distance is as follows \cite{mindensity}. 

\vspace{-2mm}
\begin{equation}
f(\epsilon) = 2\pi\rho\epsilon e^{(-\rho \pi \epsilon^2)}
\end{equation}

The probability that the distance between a randomly chosen node and its nearest neighbor is less than or equal to some distance $r$ is given by

%\footnotesize
\vspace{-4mm}
\begin{equation}
P(\epsilon<r)=1- e^{(-\rho \pi r^2)}
\end{equation}

\normalsize

If we denote $R$ as the communication range of each node, then we observe that there exists at least one unvisited node within $h$ hops (i.e., a radius of $hR$) of a token holder as long as the density of unvisited nodes satisfies:

%\footnotesize
\vspace{-4mm}
\begin{equation}
\rho_u >  (-ln(1-p))/(\pi h^2 R^2 )
\end{equation}

\normalsize

Recall from our system model that there are $N$ nodes uniformly distributed over an area $A$ and $d$ is the number of nodes per unit communication range ($\pi R^2$) of the network at any time. Thus, $N / A =  d /( \pi R^2)$.   Using this we have

%\footnotesize
\vspace{-3mm}
\begin{equation}
\rho_u  A / N >  (-ln(1-p))/(h^2 d)
\end{equation}

\normalsize

Note that $\rho_uA/N$ denotes the fraction of unvisited nodes in the network and hence we have the result. 
\end{proof}

\noindent {\bf Significance:}~~ When $p=0.95$, $d=10$ and $h=1$, $\theta /(k^2 d)=0.3$. Thus, as long as less than $70\%$ of the nodes are covered, random walks with local bias are expected to find an unvisited node within one hop 95\% of the time. With $h=2$, we see that until about $90\%$ coverage, an unvisited node can be expected within a $2$ hop neighborhood of a token. This highlights the extent to which bounded locality biased random walks can cover a significant portion of the network can be without much wasted exploration and without any supporting network structures. 

\subsection{Random walks with local bias}

\begin{theorem}
Both expected cover time and the expected number of token transfers in a random walk with local bias in a connected, mobile network of $N$ nodes with average density $d$ and a single token are $O(N(log(N/d)))$. \label{thm1}
\end{theorem}

%The expected number of unvisited neighbors in a $2$ hop range is greater than $1$ as long as $z>\theta/4d$.

\begin{proof} 
From Lemma~\ref{lemma1}, we note that the expected number of unvisited neighbors remains greater than $1$ as long as $z>\theta/d$. Thus for a fraction $(1- \theta/d)$ of the nodes, the expected distance traveled by a token is $1$.

Once the fraction of visited nodes exceeds $(1-\theta/d)$, a slowdown is expected because the token might be randomly traversing an area of already visited nodes. Now, note that there are $4d$ nodes within a $2$ hop distance of a token holder. As long as one of these nodes is unvisited, the token will find it with an expected visitation time of $4d$.  Thus for a fraction $\theta/d-\theta/4d$ of the nodes, the expected distance traveled by a token is $4d$. Continuing up to a maximum search area of $zd$, where $zd=N$, the total distance traversed by a token before visiting all nodes and the expected time for complete coverage is given by the following expression:

%\footnotesize
\vspace{-6mm}
\begin{align*}
&=O((N-\frac{N}{d})+(\frac{N}{d}-\frac{N}{4d})4d + .. + (\frac{N}{(\sqrt{z}-1)^2)}-\frac{N}{z})z) \\
&=O(N-\frac{N}{d}+N(3+5/4+ .. +(2\sqrt{z}+1)/(\sqrt{z}-1)^2 )) \\
&=O(N-N/d+ N.\sum\limits_{i=1}^{\sqrt{z}-1}(2i+1)/i^2)  \\
&=O(N+ N.\sum\limits_{i=1}^{\sqrt{z}-1}(2)/(i)+ N.\sum\limits_{i=1}^{\sqrt{z}-1} (1)/(i^2)) \\
&=O(N(1+log(\sqrt{z}-1)) ~~ \texttt{\{Euler harmonic series approx.\}} \\
&=O(N(1+log(N/d)))   \\
&=O(N(log(N/d)))   
\end{align*}

\normalsize

Thus, we have shown that both expected convergence time and the expected number of token transfers in a random walk with local bias in a connected, mobile network of $N$ nodes with average density $d$ and a single token are $O(N(log(N/d)))$.
\end{proof}

\begin{Corollary}
Both the expected convergence time and the average number of transfers per token in random walks with local bias in a connected, mobile network of $N$ nodes with average density $d$ and $k$ tokens are $O((N/k)log(N/d))$. \label{cor1}
\end{Corollary}

\begin{proof}
Note that when $k$ tokens are used for visiting all nodes, each token is on average responsible for an area of $N/k$, from which the result follows.
\end{proof}

Using the above Corollary, we observe that with $\sqrt{N}$ tokens, the expected cover time is $O(\sqrt{N}(log(N/d)))$. When $log(N)$ tokens are used, the expected convergence time is $O((N/log(N))log(N/d))$, i.e., $O(N)$ in terms of network size $N$.

\subsection{Census with gradient bias}

\begin{theorem}
Both expected cover time and the expected number of token transfers in {\em Census} with gradient bias in a connected, mobile network of $N$ nodes with average density $d$ and a single token are $O(N(1+1/d))$. \label{thm2}
\end{theorem}

\begin{proof} 
Similar to the analysis in Theorem~\ref{thm1}, for a fraction $(1- \theta/d)$  of the nodes, the expected distance traveled by a token is $1$. However, once the fraction of visited nodes exceeds $\theta/d$, the gradients will be used to pull the token towards unvisited nodes. Now, note that there are $4d$ nodes within a 2 hop distance of a token holder. As long as one of these nodes is unvisited, the token will be pulled towards that node. Thus for a fraction $\theta/d-\theta/4d$ of the nodes, the expected distance traveled by a token is 2.

Continuing up to a maximum distance of $\sqrt{N}$, the total average distance traversed by a token before visiting all nodes and the average time for complete coverage is given by the following expression:

%\footnotesize
\vspace{-6mm}
\begin{align*}
&=O((N-\frac{N}{d})+(\frac{N}{d}-\frac{N}{4d})2 + .. + (\frac{N}{(\sqrt{N}-1)^2 d)}-\frac{N}{Nd})\sqrt{N}) \\
&=O(N+\frac{N}{d}(1+1/4+1/9+...+1/N)) \\
&=O(N+\frac{N}{d}\sum\limits_{i=1}^{\sqrt{N}}(1/i^2))  \\
&=O(N(1+1/d))    
%\rlap{$\qquad \Box$} 
\end{align*}                                                     

\normalsize

Thus, we have shown that both expected convergence time and the expected number of token transfers in {\em Census} with gradient bias in a connected, mobile network of $N$ nodes with average density $d$ and a single token are $O(N(1+1/d))$.
\end{proof}

In comparison to Theorem~\ref{thm1}, we note a speed up by a factor of $log(N)$. Also, we will show later via a simulation study that the long tail in the convergence of random walk with local bias is eliminated by the use of gradients, resulting in significant speedup. 

Similar to Corollary~\ref{cor1}, we note that with $k$ tokens, the cover time for Census with gradient bias is $O(N/k(1+1/d))$. Thus, with $\sqrt{N}$ tokens, the expected cover time is $O(\sqrt{N}(1+1/d))$. When $log(N)$ tokens are used, the expected cover time is $O(N/log(N)(1+1/d))$.

%However, when k tokens are used for visiting all nodes, the maximum size of gradients reduces by a factor of k.
\begin{theorem}
The expected gradient message overhead in Census with gradient bias in a connected, mobile network of $N$ nodes with density $d$ and $k$ tokens is $O((N/k)log(N/d))$. \label{thm3}
\end{theorem}

\begin{proof}
Following the lines of Theorem~\ref{thm1}, we note that for a fraction $(\theta/d-\theta/4d)$ of the nodes, the average gradient set up cost will be $4d$.  And, for a fraction $(1/((\sqrt{p}-1)^2 d)-1/pd)$ of the nodes, the gradient setup cost will be $pd$, where $pd=N/k$. The result then follows from summing up the series as shown in the proof of Theorem $4.2$. 
\end{proof}

Thus, we note that when using gradient bias, there is an extra overhead to pull the tokens towards unvisited nodes, but this is compensated by reduction in the number of required token transfers and reduction in convergence time. Moreover, the gradient message overhead decreases linearly with number of tokens.

\iffalse
\subsection{Utilizing concurrent independent trials}

To reduce cover time of {\em Census} without introducing gradient message overhead, we explore the idea of using multiple, independent trials of {\em Census} with local bias. Each trial may have one or more tokens. Since these trials are independent of each other with tokens initiated for each trial at random points in the network, the set of nodes visited by a given trial is independent of the set of nodes visited in another trial. Thus, if $p_n$ is the probability that a node is not visited by a trial, the probability that the same node is not visited by $c$ independent trials is $p_n^c$. We utilize this fact to state the following Lemma. We verify the independence and diversity of trials using simulations in Section~$5$.

\begin{lemma}
Given $c$ independent trials of {\em Census} with local bias, each terminated at $m\%$ coverage, the union of the set of nodes visited by each of the $c$ trials yields a coverage of $(1-(1-m)^c)$.
\end{lemma}
\fi

\subsection{Comparison with other techniques for aggregation}
\label{sec:comp}

\subsubsection{Structured data aggregation} 

The problem of computing aggregates in ad-hoc networks is a well-studied one, especially for static sensor networks. Solutions include Directed Diffusion \cite{directdiff}, Collection Tree Protocol \cite{ctp}, Sprinkler \cite{naik_sprinkler} and TAG \cite{tag}. However, in mobile networks and networks with frequent link changes, topology driven structures are likely to be unstable and incur a high communication overhead for maintenance.  It has been observed that routing protocols for MANETs such as OLSR \cite{rtsw3} are unable to scale beyond $100$ - $150$ nodes \cite{rtsw1}. In a recent paper \cite{rtswarxiv}, we analyzed reasons for this {\em scaling wall} in MANETs. Firstly, as network size increases, paths are likely to fail more often - the median path connectivity interval in the network falls as $O(\sqrt{N})$, where $N$ is the network size. Secondly, we observe that structure based approaches are not robust to motion models and node speeds because even small changes to node speed significantly increase the frequency of link changes and hence increase the path failure rate and instability. Finally, in order to be successfully used for routing, the broken paths need to be fixed much faster than the rate at which they break. This involves fast link estimation for discovering broken paths and then repairing those paths, both of which incur high message overhead. Not surprisingly, none of these protocols have been successfully adapted or evaluated for MANETs. On the other hand, we observe that {\em Census} is able to scale to several thousand of nodes in highly mobile networks. We also find that its performance is largely unaffected by motion model and node speeds (as shown later in Fig.~\ref{res7} and Fig.~\ref{speedtrialsab}). 

\subsubsection{Flooding and Gossip}

A structure free approach such as flooding data from all nodes to every other node has a messaging cost of $O(N^2)$, and is not any faster than {\em Census}. Alternatively, for problems such as average consensus, one could use multiple rounds of local gossip where in each round a node averages the current state of all its neighbors and this procedure is repeated until convergence \cite{gossip-survey, randgossip}. However, this method requires several iterations and has also been shown to have a communication cost and completion time of $O(N^2)$ for convergence in grids or random geometric graphs, where connectivity is based on locality \cite{rabbatgossip}. Census has a communication cost of $O(Nlog(N)/k)$, and it can be used for applications beyond just averaging.

A variant of flooding and gossip for general aggregation problems is a {\em diffusion}-like approach where the initiating node broadcasts an $N$-bucket register (one for each node in the network). Each time a node receives this register, it adds its own state into the register (if it's not already added) and rebroadcasts the register if it learned about any new node in this iteration. This process continues until all nodes have complete copies of the $N$-bucket register. Note that the size of each message in this technique is $O(N)$ and therefore the messaging cost is at least $\Omega(N^2)$. Moreover, this technique assumes that the ids of all the nodes in the network are known a priori. On the other hand, {\em Census} does not assume any knowledge about the nodes in the network and has a much lower communication cost.

\section{Evaluation}
In this section, we quantitatively evaluate the convergence characteristics and performance of {\em Census} under different network conditions. More specifically, in Section~\ref{sec:conv}, we quantify the convergence characteristics of {\em Census} and compare with those of pure random walks and random walks with local bias; our results characterize the improvement obtained by the gradient bias in {\em Census} by mitigating the long tail in cover time. And in Section~\ref{sec:prop}, we evaluate the message overhead and cover time for {\em Census} under different network conditions such as densities, number of tokens, mobility models, and speeds, and compare this with random walks that have only local biasing. 

For simulations of {\em Census} using ns-3, we set up MANETs ranging from 125 to 4000 nodes with varying number of tokens that are initiated at random location within the network. Nodes are deployed uniformly in the network and the deployment area and communication range are chosen such that the average neighborhood size ($d$) remains constant irrespective of network size. We test with $d=7$, $10$, and $13$ in our simulations with the following mobility models: 2-d random walk, random waypoint and Gauss-Markov. The average node speeds range from $3$ to $15$ m/s.

\subsection{Convergence characteristics}
\label{sec:conv}

\begin{figure*}[htbp]
\vspace*{-3mm}
  \begin{center}
    \mbox{
      \subfigure[] {\scalebox{0.35}{\includegraphics[width=\textwidth]{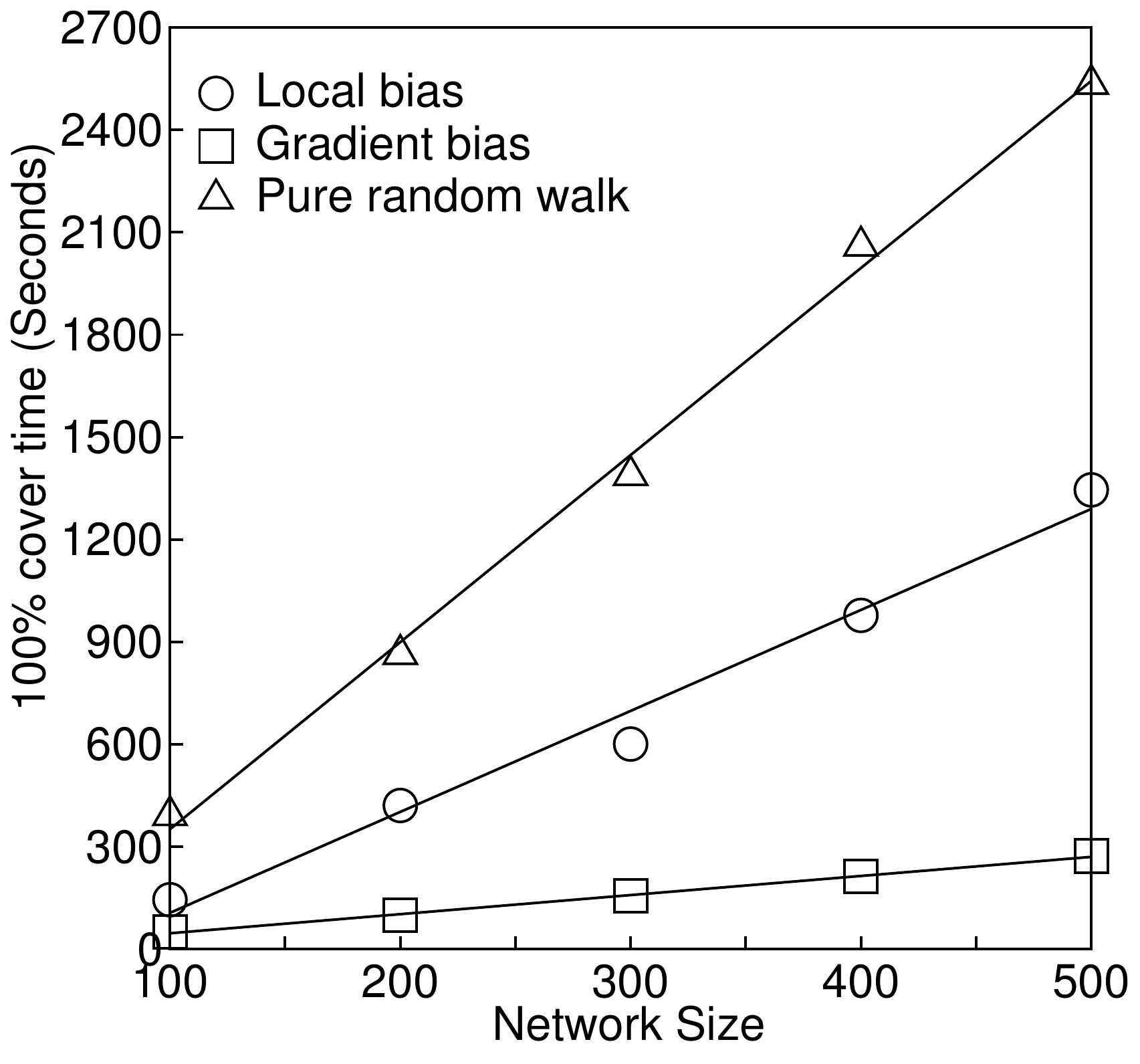}} \label{fig:1a}} \quad
      \subfigure[] {\scalebox{0.35}{\includegraphics[width=\textwidth]{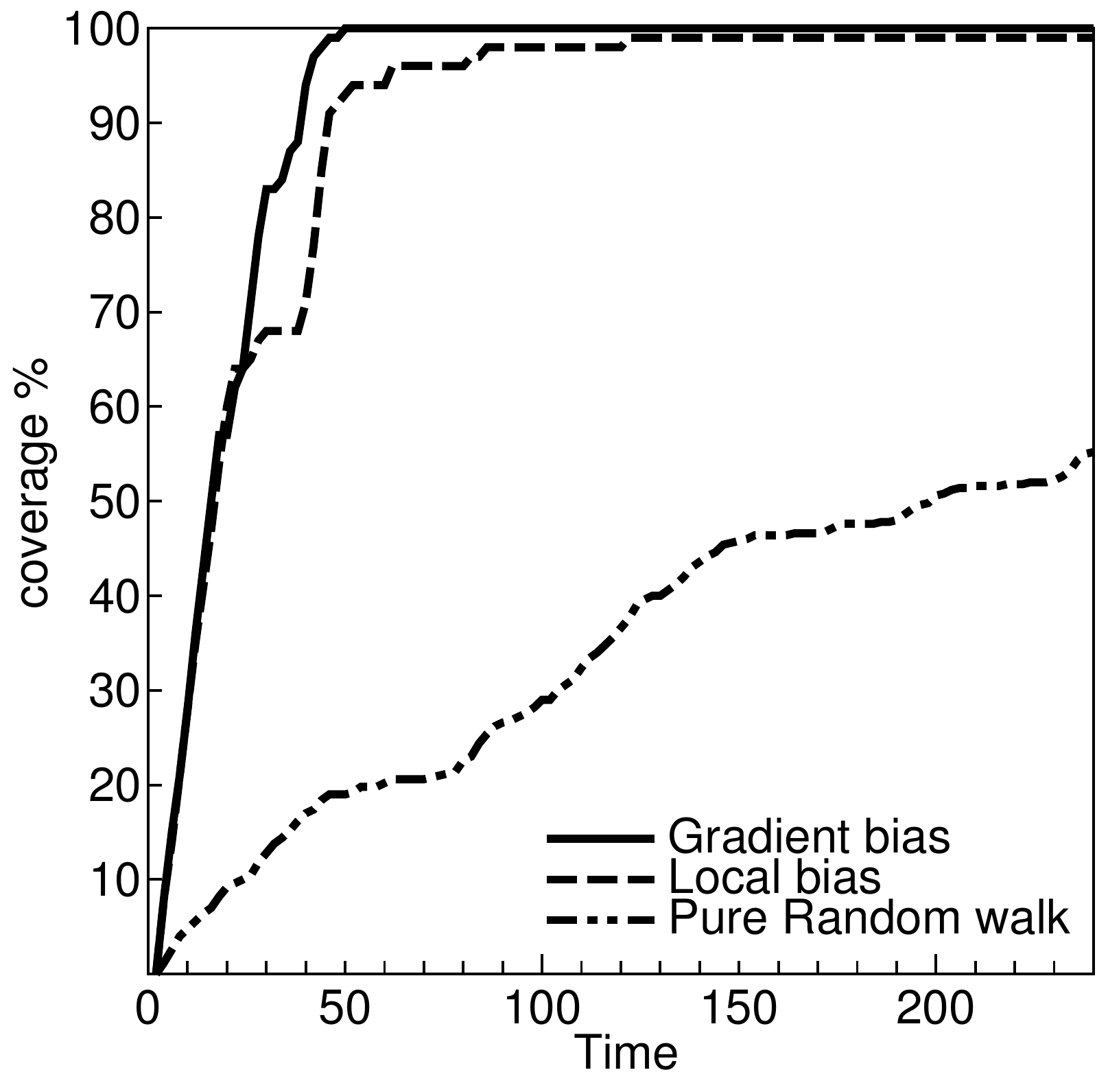}} \label{fig:1b}} \quad
      
      }
    \vspace{-2mm}  
    \caption{(a) Cover time vs network-size (single token) (b) Convergence pattern on a network of 500 nodes for a given trial with and without bias }
       \label{res1}
  \end{center}
\vspace*{-2mm}
\end{figure*}

\subsubsection{Comparison of pure, local bias and gradient bias random walks}

Here, we compare the convergence characteristics of pure random walks with that of local and gradient bias. We have used network sizes of $100-500$ and a single token in a random 2D-walk mobility model where the nodes move in a certain direction for a fixed distance and then choose a new random direction. The $100\%$ cover time are shown in Fig.~\ref{fig:1a}. We observe that local biasing is about two times faster than pure random walks and  gradient bias is about $12-14$ times faster than pure random walks.  

Notably, local biasing is actually much faster than pure random walks from the beginning until a major portion of the network is covered, but a slowdown happens towards the tail. This difference in convergence characteristics is illustrated in Fig.~\ref{fig:1b}, where we show the convergence pattern in a single randomly chosen trial of $500$ nodes for pure random walk, local and gradient biased Census. In this particular trial, we observe that until around the $65\%$ mark, local bias proceeds on par with gradient bias, but then slows down slightly. This is because, until this point local biasing enables a token to find an unvisited node directly and there are very few wasted explorations. A more pronounced slowdown for local biasing is noticed around $80\%$, whereas a pure random walk is slow throughout. After this point, there are many wasted explorations in finding unvisited nodes when only local biasing is used, whereas {\em Census} with gradient bias proceeds at a uniform rate throughout. 

\begin{figure*}[htbp]
\vspace*{-3mm}
  \begin{center}
    \mbox{
      \subfigure[] {\scalebox{0.35}{\includegraphics[width=\textwidth]{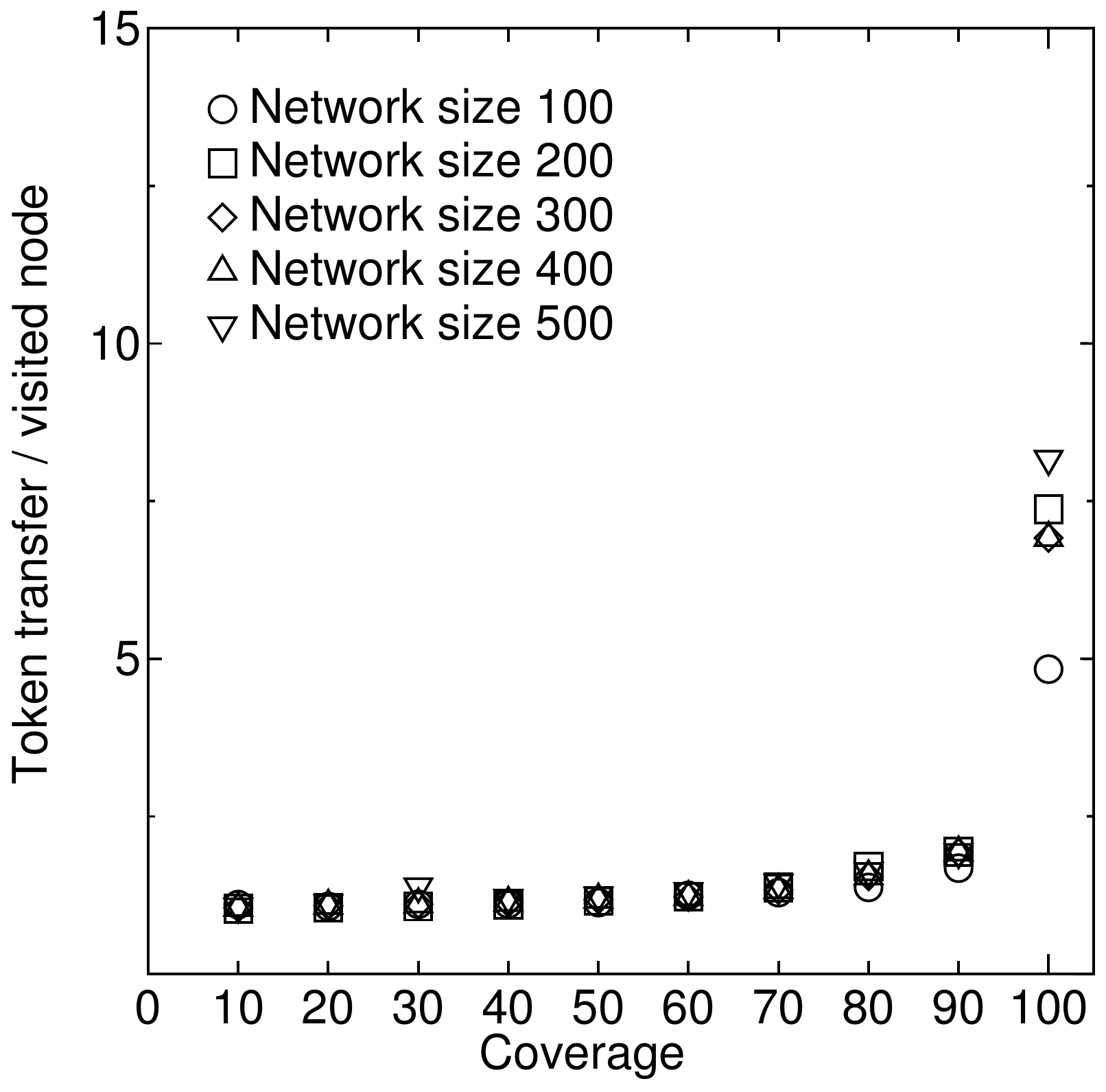}} \label{fig:2a}} \quad
      \subfigure[] {\scalebox{0.35}{\includegraphics[width=\textwidth]{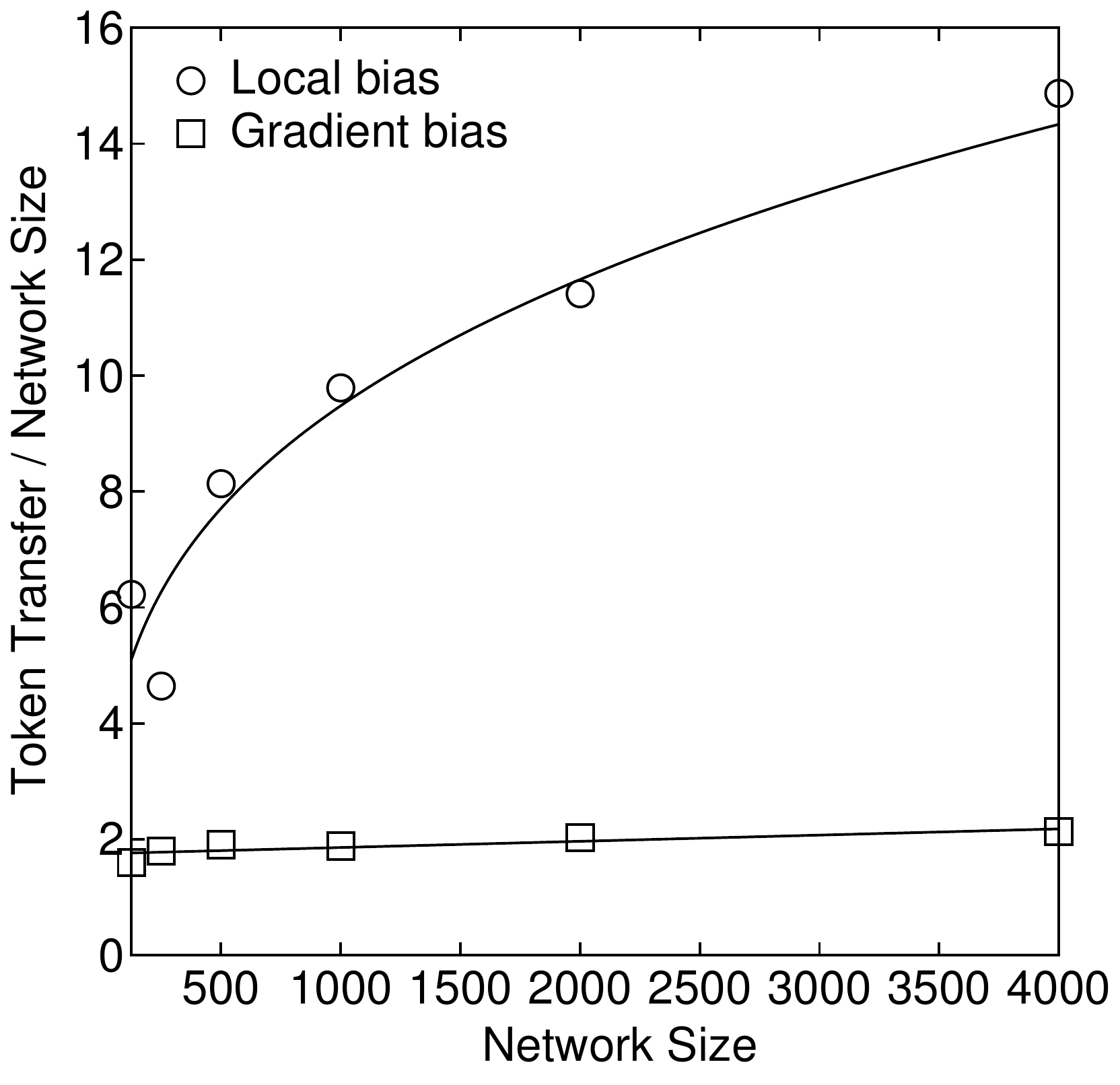}} \label{fig:4}} \quad
      
      }
    \vspace{-2mm}  
    \caption{(a) Token passing overhead as a function of coverage percentage (local bias) (b) Exploration overhead vs network size }
       \label{res2}
  \end{center}
\vspace*{-2mm}
\end{figure*}

We analyze this effect of {\em exploration overhead} for random walks with local bias in more detail in Fig.~\ref{fig:2a} which plots the ratio of token transfers to the number of visited nodes at different stages of coverage for random walks with local bias (averaged over multiple trials). This ratio captures the wasted explorations where a token is repeatedly passed to already visited nodes. We see that with local bias, the ratio stays low ($<2$) until about the $70-80\%$ mark and then starts rising rapidly.  The observation matches our result in Lemma~\ref{lemma1}. But the token passing overhead for {\em Census} remains close to $2$ throughout without any sharp rise. 

%Another observation from Fig.~\ref{fig:2a} is that the exploration overhead ratio for local biasing stays constant {\em irrespective of network size} until a certain threshold, after which there is also a sharply increasing variance. This observation is useful in specifying terminating conditions for random walks with local bias when a required fraction of nodes have been visited.

\subsubsection{Exploration overhead as function of network size}

In Fig.~\ref{fig:4}, we compare the ratio of token transfers to the visited nodes at different network sizes with $100\%$ coverage. We observe that for local bias, this ratio grows as $log(N)$, while for the gradient bias it is almost flat, matching our results in Theorem~\ref{thm1} and Theorem~\ref{thm2}. Redundant token passes are very low with gradient bias and the ratio stays close to $2$ even for a $4000$ node network. The exploration overhead factor indicates the number of token passing transactions required for coverage and is therefore a more accurate representation of convergence characteristics than the absolute cover time which is quite implementation specific. For instance, in our implementation each transaction (i.e., each iteration of token announcement, token requests and token passing) took on average $250ms$. But this number could be much smaller using methods such as \cite{collabcounting} that use collaborative communication for estimating neighborhood sizes that satisfy given predicates.

%In Fig.~\ref{fig:2b}, we compare the message overhead of local and gradient bias. Token messages include the announcement, request and hand-off related messages. The gradient related message overhead grows almost linearly with network size. But this is compensated by a significant reduction in the required token messages in the gradient bias scheme as compared to the local bias scheme.Further recall from Fig.~\ref{fig:2a} that even with local bias, this ratio stays under $2$ until around $70\%$ and only then rises steeply. 
\subsection{Census properties under different network conditions}
\label{sec:prop}

In this section, we evaluate the message overhead and cover time for Census (as well as local bias) under different densities, number of tokens, mobility models, and speeds.

\subsubsection{Message overhead}

\begin{figure}[t]
  \begin{center}
    \includegraphics[width=.35\textwidth]{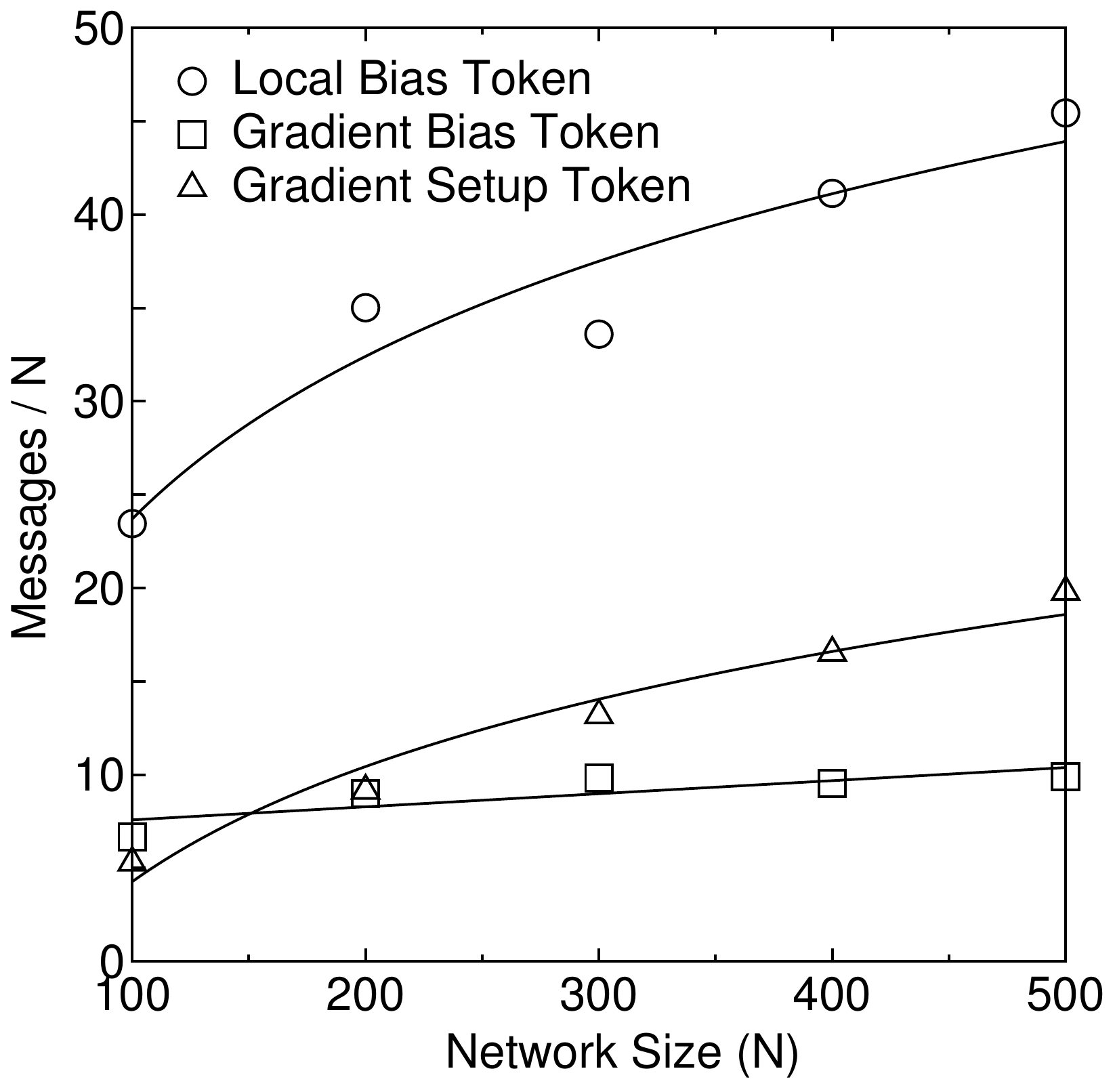}
    \caption{Analysis of message overhead for Census and local bias}
    \label{fig:2b}
  \end{center}
\end{figure}

In Fig.~\ref{fig:2b}, we compare the message overhead of Census and local bias. To highlight the log factor overhead, we show messages normalized by network size on y-axis and network size on x-axis. The values for token messages with gradient bias stay constant indicating that the token messages grow linearly with network size. The curve for local-bias token messages and the gradient setup messages grow as $log(N)$, indicating a $Nlog(N)$ messaging overhead. Note that the gradient setup cost is compensated by a significant reduction in the required token messages in Census as compared to the local bias scheme. These results show that Census achieves superior performance both in terms of cover time and message overhead over that of local bias, despite the use of short multi-hop gradients. 

\subsubsection{Multiple tokens}

First, we quantify the impact of using $\sqrt{N}$ and $\log(N)$ tokens. Fig.~\ref{fig:3a} shows the cover time with $\sqrt{N}$ tokens in the network for Census and local bias. The network sizes that we simulate are 125, 250, 500, 1000, 2000 and 4000. The corresponding number of tokens used in the network is 11, 15, 22, 31, 42 and 62 respectively. We observe that the coverage time grows only as $O(\sqrt{N})$, matching our analysis. In Fig.~\ref{fig:3b}, we show the impact of using $log_2(N)$ tokens. The network sizes that we simulate are 125, 250, 500, 1000, 2000 and 4000. The corresponding number of tokens used in the network is 7, 8, 9, 10, 11 and 12 respectively. Here, the trend is observed to be linear. 

\begin{figure*}[h]
\vspace*{-3mm}
  \begin{center}
    \mbox{
      \subfigure[] {\scalebox{0.35}{\includegraphics[width=\textwidth]{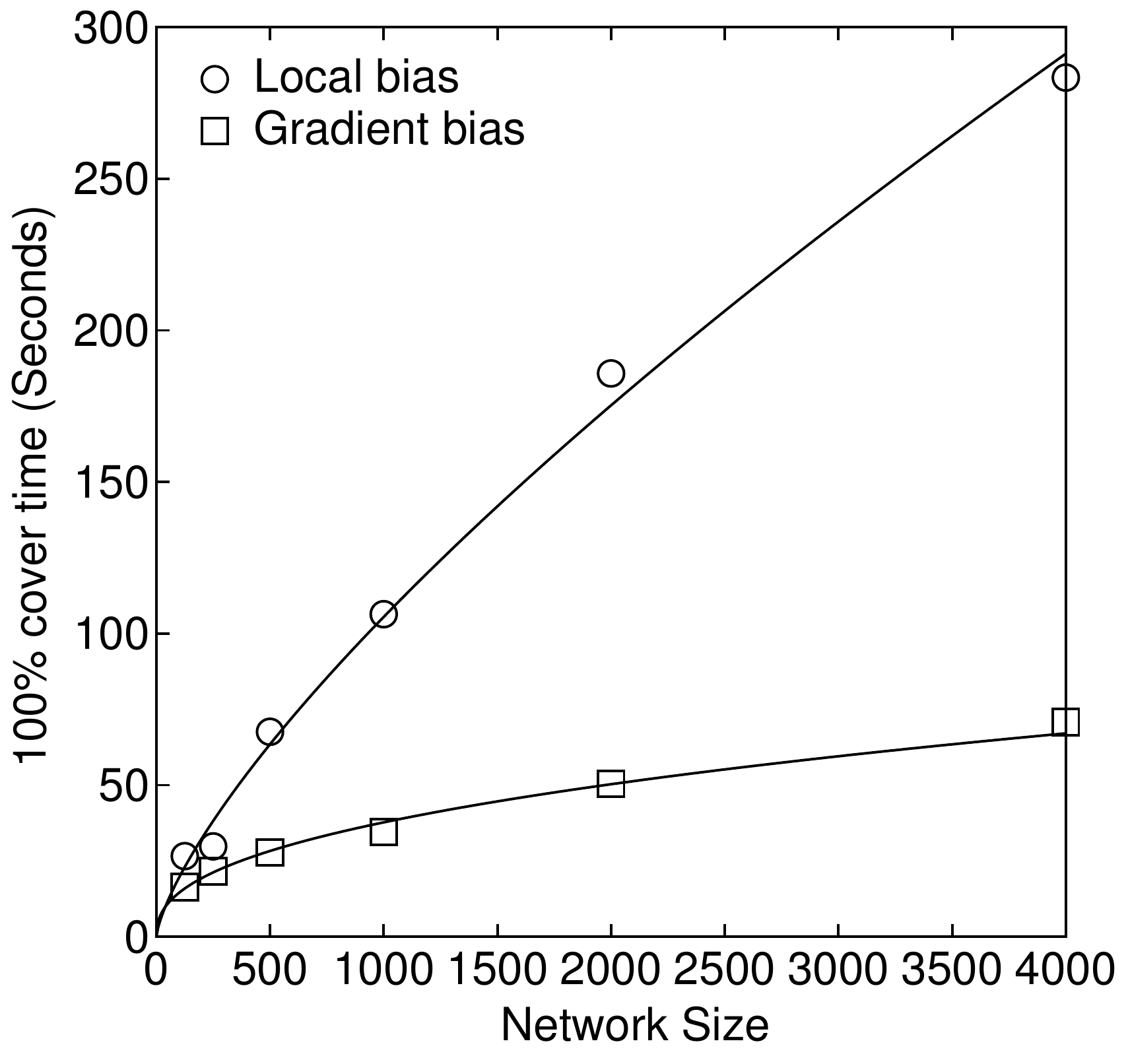}} \label{fig:3a}} \quad
      \subfigure[] {\scalebox{0.35}{\includegraphics[width=\textwidth]{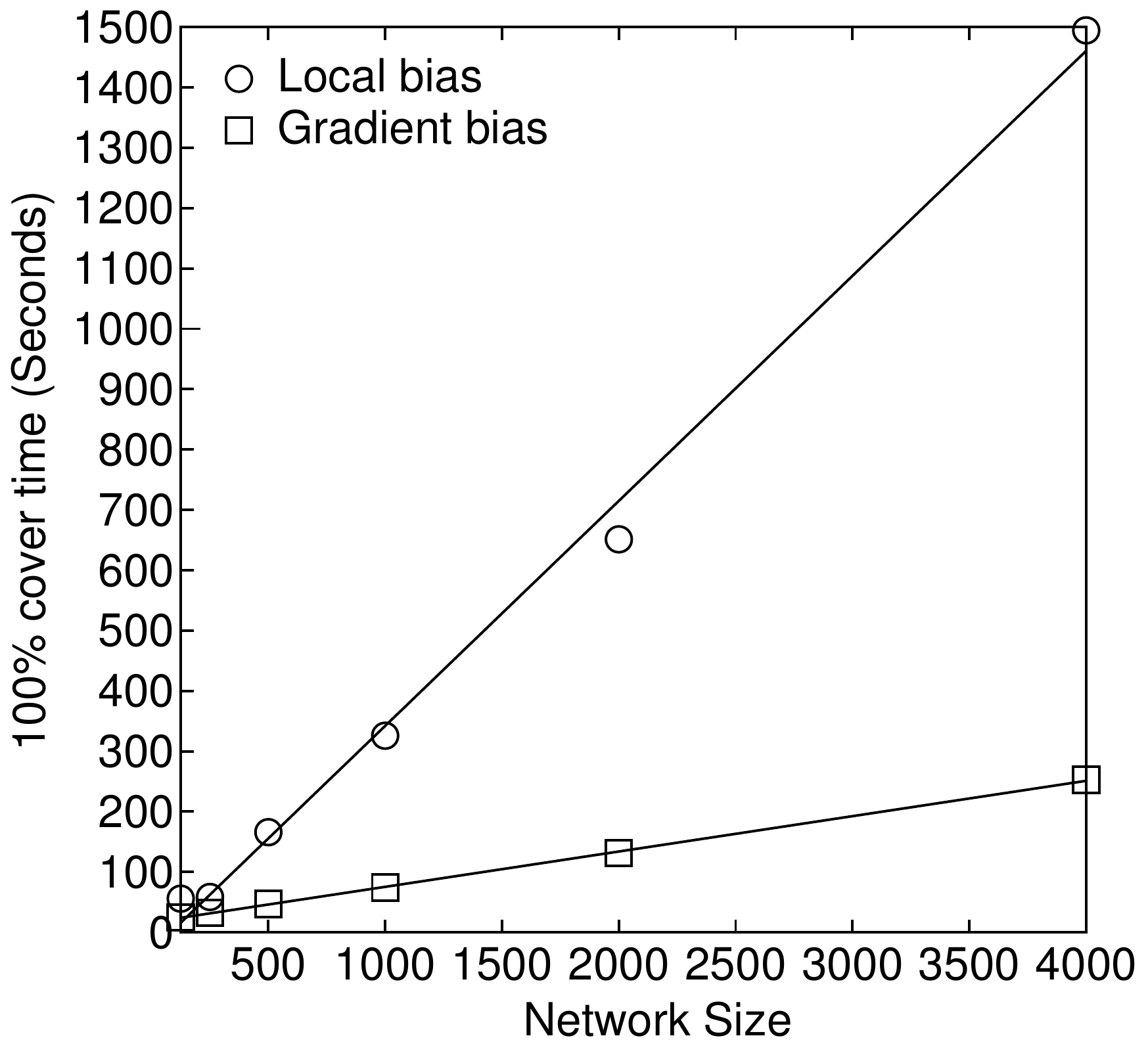}} \label{fig:3b}} \quad
      
      }
    \vspace{-2mm}  
    \caption{(a) Cover time vs network-size ($\sqrt{N}$ tokens) (b) Cover time vs network-size ($\log{N}$ tokens) }
       \label{res3}
  \end{center}
\vspace*{-2mm}
\end{figure*}

\begin{figure*}[h]
\vspace*{-3mm}
  \begin{center}
    \mbox{
      \subfigure[] {\scalebox{0.32}{\includegraphics[width=\textwidth]{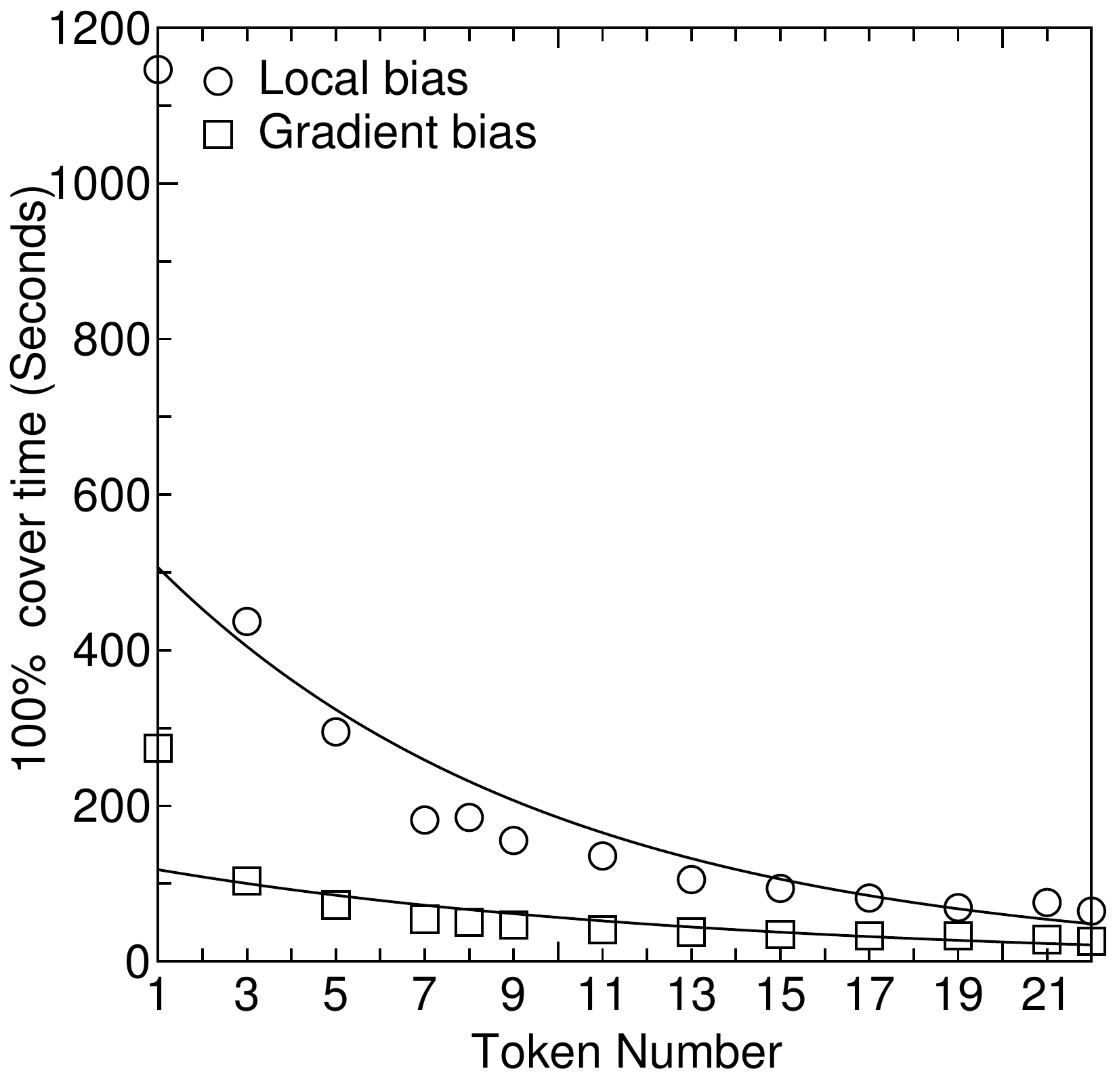}} \label{fig:5a}} \quad
      \subfigure[] {\scalebox{0.38}{\includegraphics[width=\textwidth]{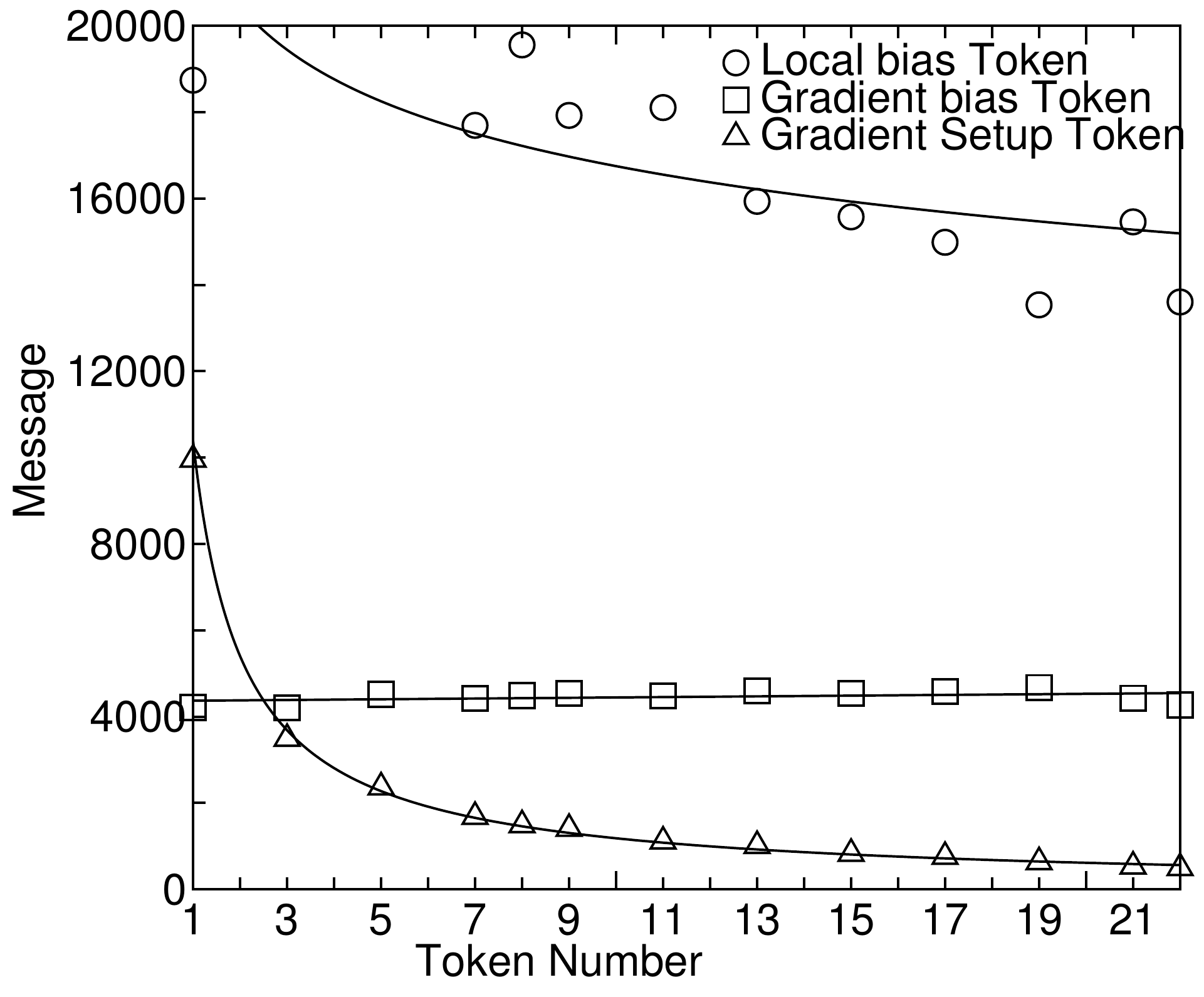}} \label{fig:5b}} \quad
      
      }
    \vspace{-2mm}  
    \caption{(a) Cover time vs number of tokens ($N=500$) (b) Messages vs number of tokens ($N=500$)}
       \label{res5}
  \end{center}
\vspace*{-2mm}
\end{figure*}

Next, in Fig.~\ref{fig:5a}, we analyze the cover time as a function of number of tokens. The network size is $500$ and the number of tokens are varied from $1$ to $22$. We notice that the cover time decreases linearly with the number of tokens, matching our analysis. In Fig.~\ref{fig:5b}, we compare the messaging overhead for the same scenario. For Census, we see that the gradient message overhead decreases linearly with the number of tokens. The token passing overhead for gradient bias stays roughly constant, because there is very little overhead to begin with. The token passing overhead for local bias decreases linearly with number of tokens.

%To better understand the trend, we show a log-log plot. The x-axis in Figure 5(b) is log(N). The y-axis shows the values of $log(y.log(N))$, where y is the coverage time. The graph is observed to have a slope of approximately 1, showing that $y = N/log(N)$, matching our analysis.

\subsubsection{Impact of density}

\begin{figure*}[htbp]
\vspace*{-3mm}
  \begin{center}
    \mbox{
      \subfigure[] {\scalebox{0.37}{\includegraphics[width=\textwidth]{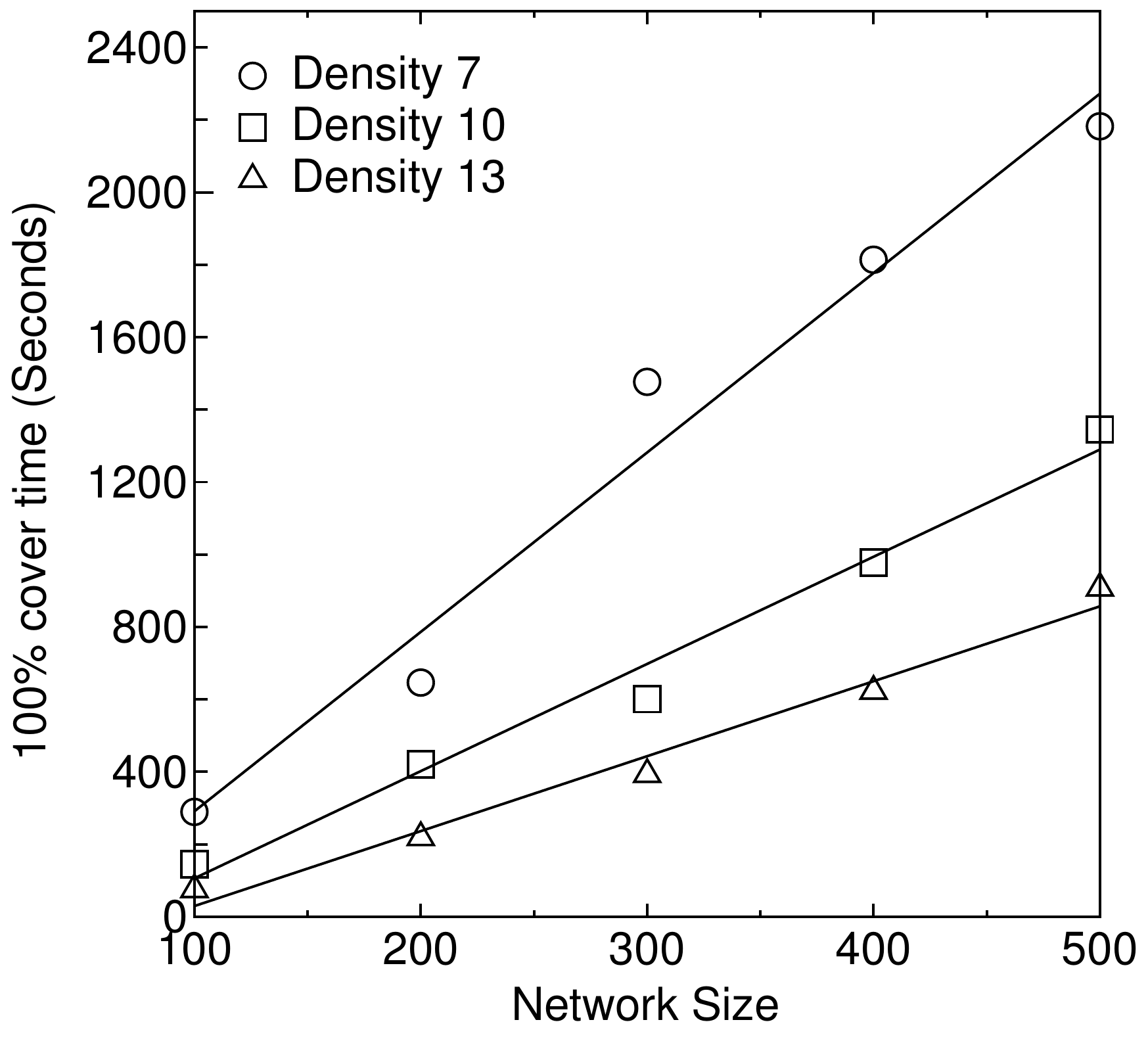}} \label{fig:6a}} \quad
      \subfigure[] {\scalebox{0.35}{\includegraphics[width=\textwidth]{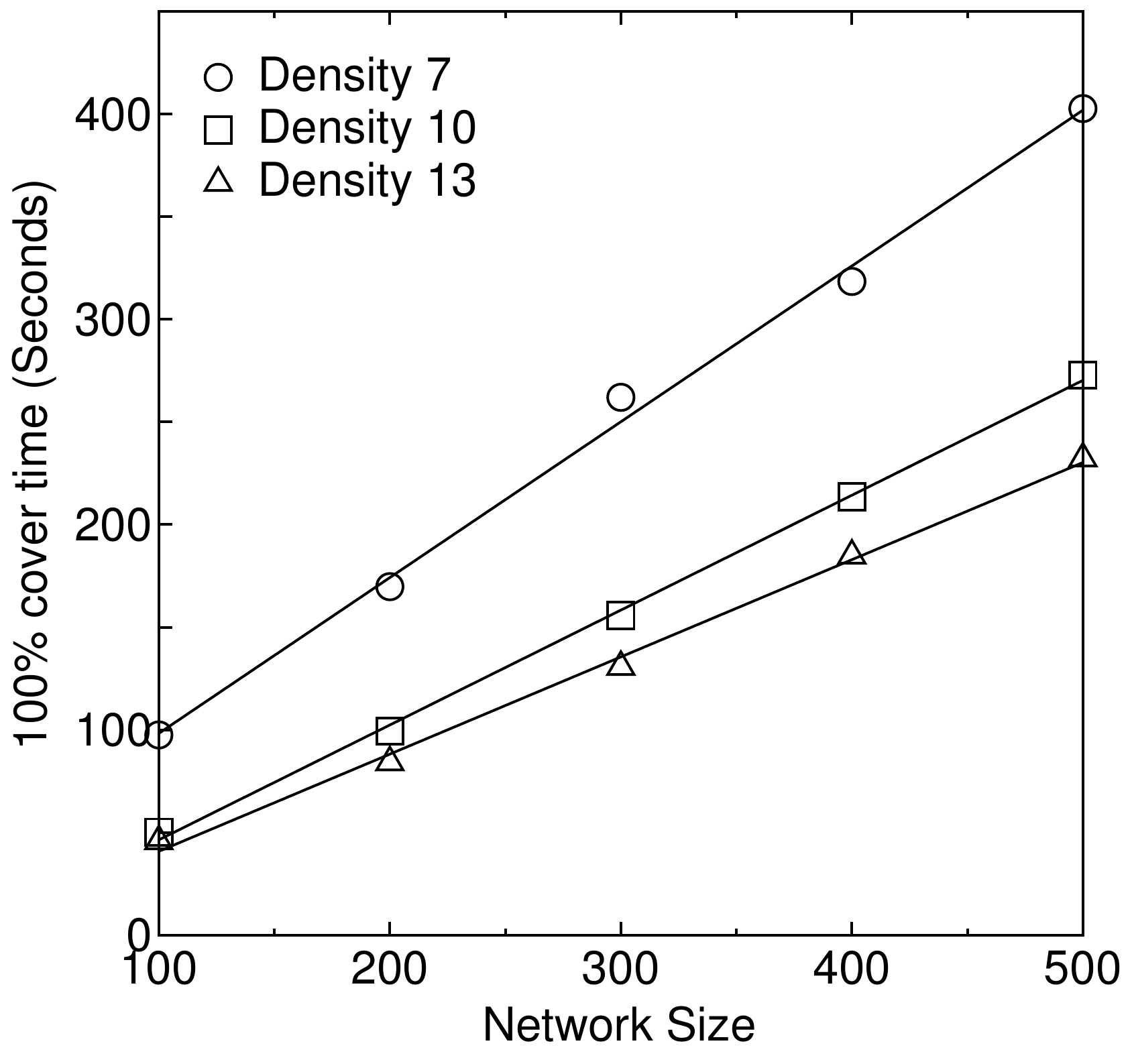}} \label{fig:6b}} \quad
      
      }
    \vspace{-2mm}  
    \caption{(a) Impact of density (local bias) (b) Impact of density (gradient bias)}
       \label{res6}
  \end{center}
\vspace*{-2mm}
\end{figure*}

In Fig.~\ref{res6}, we quantify the impact of the average neighborhood size ($d$) on cover time. We have considered single token scenario to eliminate effects of multiple tokens. The graph shows an expected decrease in cover time at higher densities.

\subsubsection{Impact of mobility model}

\begin{figure*}[htbp]
\vspace*{-3mm}
  \begin{center}
    \mbox{
      \subfigure[] {\scalebox{0.37}{\includegraphics[width=\textwidth]{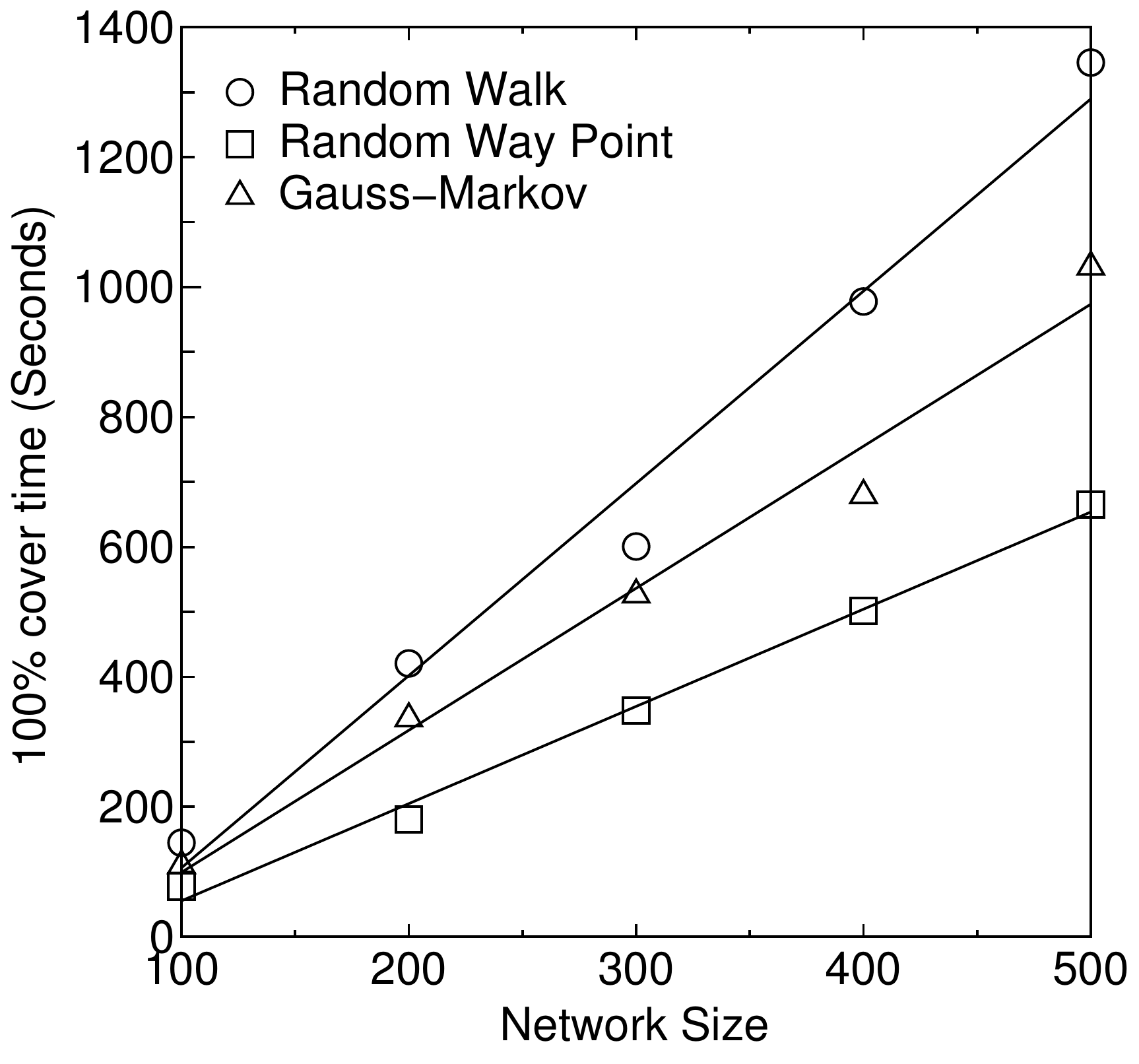}} \label{fig:7a}} \quad
      \subfigure[] {\scalebox{0.35}{\includegraphics[width=\textwidth]{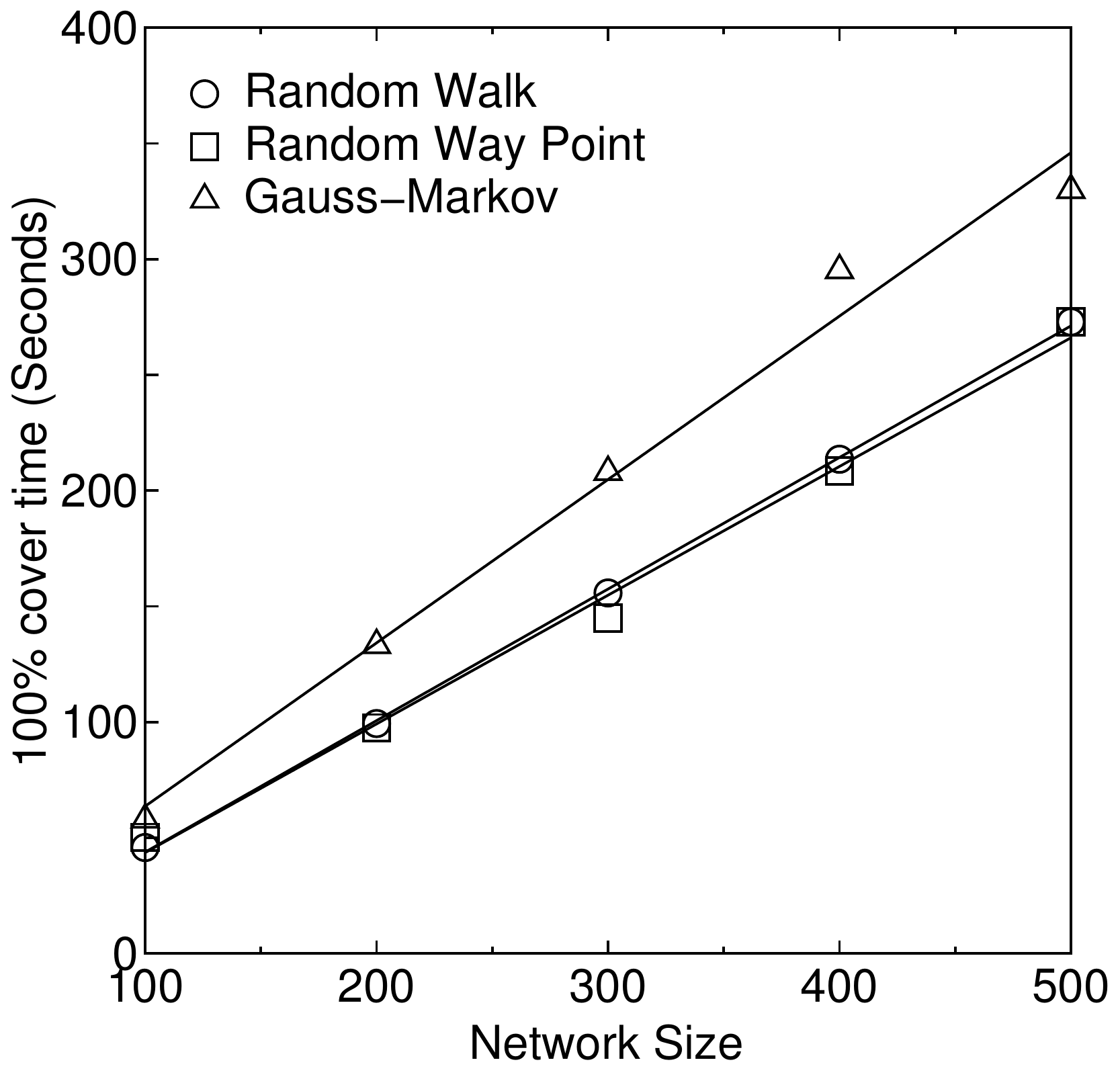}} \label{fig:7b}} \quad
      
      }
    \vspace{-2mm}  
    \caption{(a) Impact of mobility models (local bias) (b) Impact of mobility models (Census) }
       \label{res7}
  \end{center}
\vspace*{-2mm}
\end{figure*}

We simulated {\em Census} under other mobility models, random waypoint and Gauss Markov. The node speeds remain in the range of $2-4$ m/s. For random way point, the pause time is set to $2$ seconds between successive changes. In the Gauss Markov model, where motion characteristics are correlated with time, tuned with a parameter $\alpha$. We set $\alpha=0.75$. Velocity and direction are changed every $1$ second in the Gauss Markov Model. From Fig.~\ref{res7}, we observe that the cover time characteristics are similar for these different mobility models, even if the exact times show some variation, indicating robustness with respect to mobility model especially for Census. 

\subsubsection{Impact of node speed}

\begin{figure*}[htbp]
\vspace*{-3mm}
  \begin{center}
    \mbox{
      \subfigure[] {\scalebox{0.37}{\includegraphics[width=\textwidth]{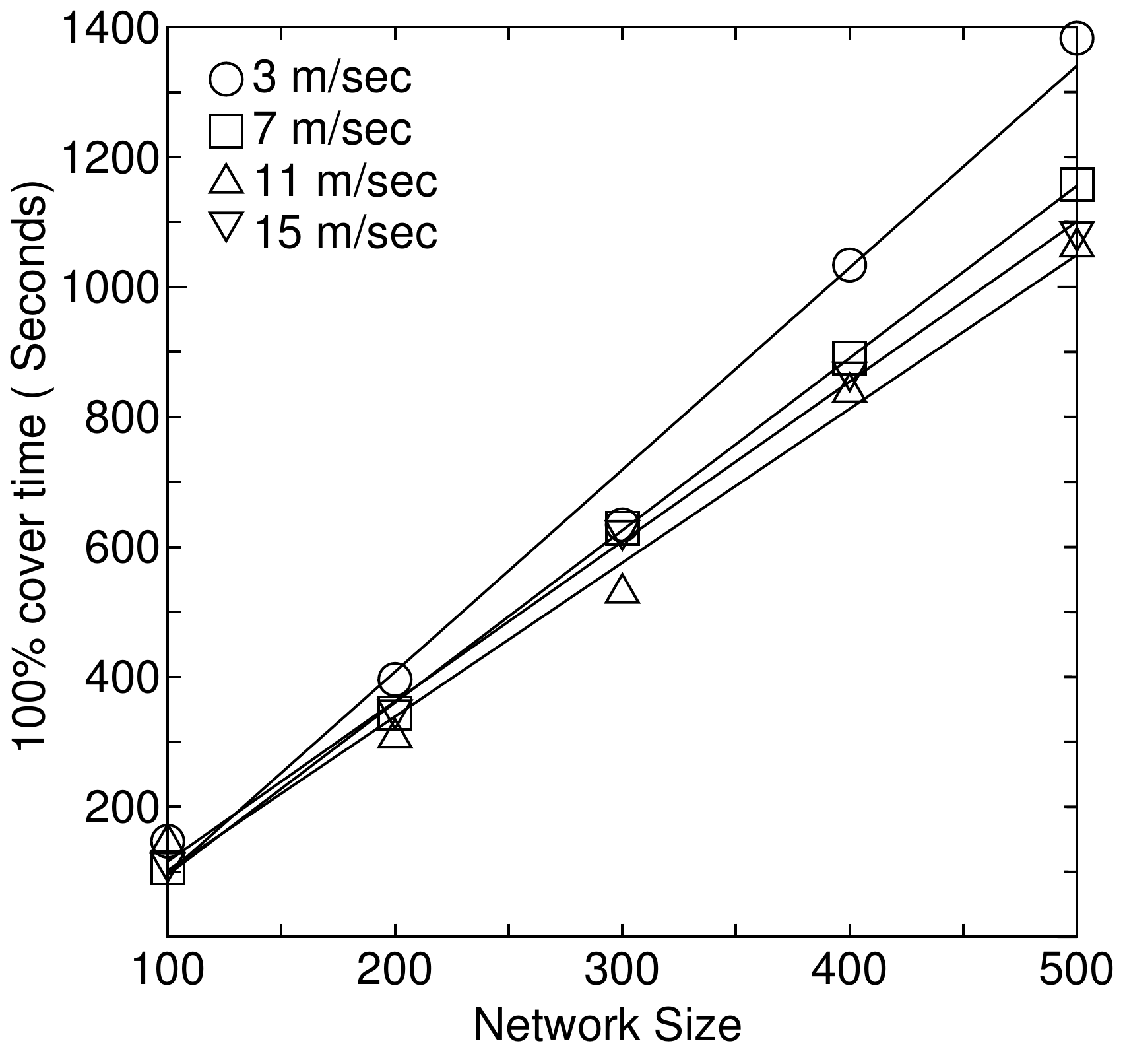}} \label{fig:speedtrials}} \quad
      \subfigure[] {\scalebox{0.35}{\includegraphics[width=\textwidth]{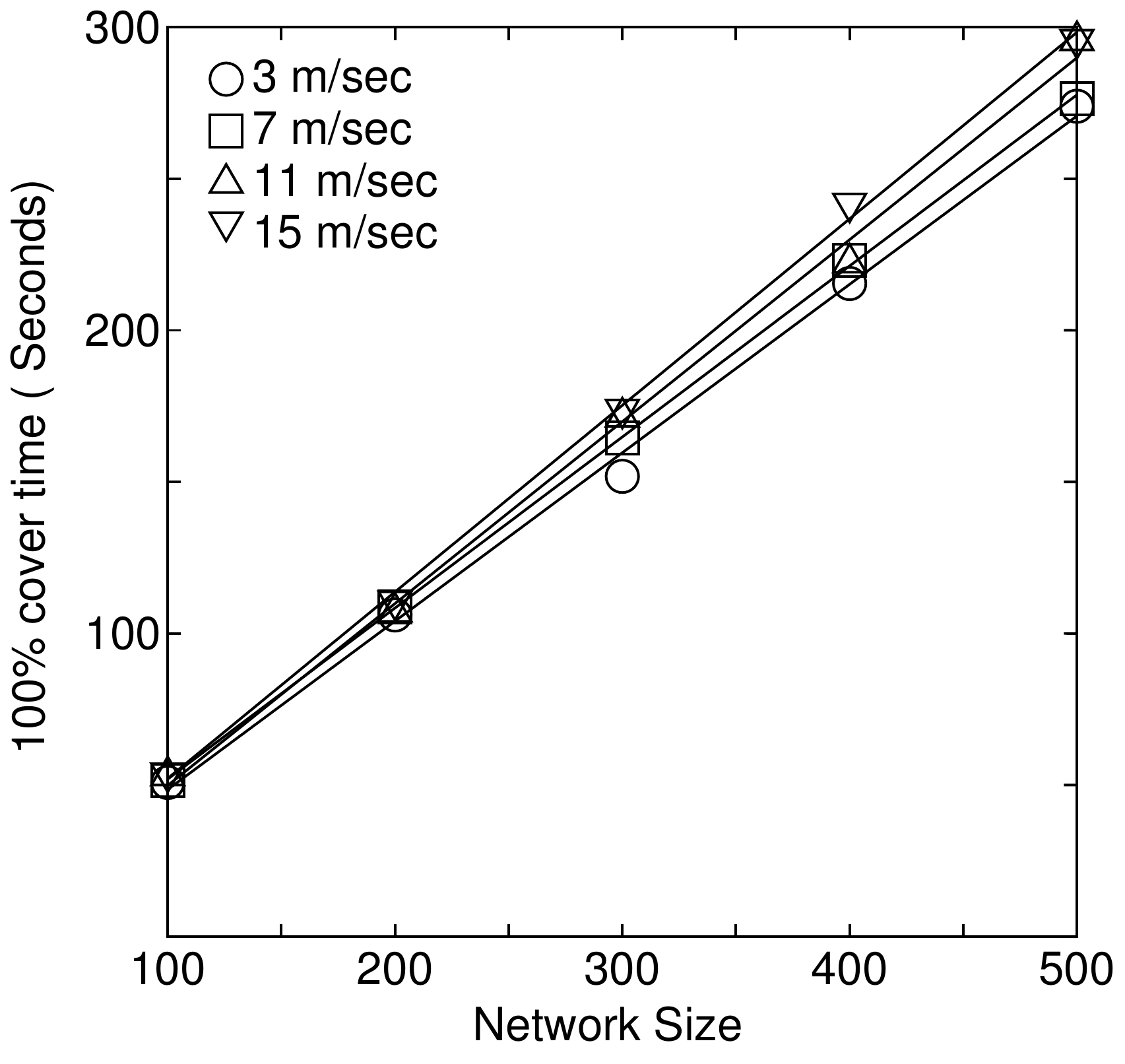}} \label{fig:speedtrials-gradient}} \quad
      
      }
    \vspace{-2mm}  
    \caption{(a) Impact of node speed (local bias) (b) Impact of node speed (gradient bias) }
       \label{speedtrialsab}
  \end{center}
\vspace*{-2mm}
\end{figure*}

To highlight that {\em Census} is robust with respect to the rate of mobility induced link changes, we simulated {\em Census} under different node speeds. As observed in Fig.~\ref{speedtrialsab}, even as the average node speed increases to $15m/s$, the cover time remains quite steady. For random walks with local bias, we observe that for larger networks, cover time improves with higher node speed. This is probably explained by the fact that at higher speeds, nodes which are temporarily disconnected from the portion that has a token, tend to converge with the connected component sooner and thus reduce the long tail in cover time.

\section{Discussion}
In this section, we discuss some issues and optimization techniques that are not related to the core idea of using random walks for token coverage, but nevertheless are important in the context of implementing {\em Census} in a MANET.

\subsection{Reliable token transfer}

Reliable token transfer is critical for successful operation. If a token is released by a node, but the intended recipient did not receive the token reply message, the token is lost. At the same time, if the sending node relies on acknowledgements to release a token, it is possible that the acknowledgements are lost and duplicate tokens are created. For applications where duplicate counting is not permitted, this is a problem. This issue can be addressed in practice by using acknowledgments in conjunction with checkpoints. The procedure is described below.

As soon as a token reply has been sent, the sender releases the token (the node resets holder to zero). At the same time, it remains in a waiting state for acknowledgements from the recipient. If an acknowledgement is not received within a time $T_a$, the token send message is repeated up to a maximum of $K$ re-tries. If the recipient receives the token multiple times, it simply repeats the acknowledgement message. However, if the token sender does not receive the acknowledgement even after $K$ retries, it creates a {\em checkpoint} for the token: (a) the aggregate computed thus far is appended to the token along with the token id, (b) a fresh token id is created (unique ids can be created by simply assigning a node's id to the token during creation) and (c) the token aggregate is reset. It is possible that the token was actually successfully passed, but even in this case the checkpoint will not create duplicate counting. At the same time, the process ensures that data is not lost. 

\subsection{Termination detection}

When using Census, termination can be deterministically detected. Note that when all nodes have been visited, the gradient setup will be terminated because the gradient setup is only initiated by nodes that have not been visited. As a result, a node that holds a token will continue to get only a level $0$ reply for its token announcement. If a gradient is being setup, there would be at least one neighboring node with a value of level $> 0$. Therefore, when nodes holding the token get a level $0$ reply from all its neighbors over an interval greater than the gradient refresh interval, the holder nodes can conclude that all nodes in the network have been visited.

In contrast, when using random walks with only locally biased random walks, there is no deterministic way to detect termination. However, as the percentage of visited nodes increases, the ratio of token transfers to the visited nodes starts increasing. This ratio may be used to design an approximate threshold for termination detection. Moreover, the result in Fig.~\ref{fig:2a} shows the expected ratio of token transfers to the visited nodes at different levels of coverage with local bias. In this figure, we see that until about $80-90\%$ coverage, there is very little variance in the token passing overhead ratio across different network sizes. Therefore, these values can be used to determine approximate thresholds for terminating a random walk trial at a desired level of coverage, irrespective of network size. 

\subsection{Token exfiltration}

Once the initiated tokens have visited all nodes, it is necessary to ex-filtrate the tokens to a given location such as the operating base station or to one or more querying nodes in the network. Instead of using structures to route these aggregates towards querying nodes, a simple solution is to simply flood the aggregate tokens across the network in $O(D)$ time (where $D$ is the network diameter) with an $O(Nk)$ message overhead where $k$ is the number of tokens. This leads to a potential question: why not use flooding or diffusion based approaches all the way? Note that the cost of disseminating data from each node to all other nodes is $O(N^2)$ where $N$ is the number of nodes in the network.  By using a fixed number of $k$ tokens to first compute the aggregates and then flooding the aggregates, the message overhead for flooding is only $O(Nk)$.  Thus, our bounds on message overhead remain unaffected. 

Note that other structure-free solutions are also possible for token ex-filtration. For instance, another potential solution is to transmit the $k$ aggregated tokens using a long distance transmission link (such as cellular or satellite links) in hybrid MANETs where the {\em long links} are used for infrequent, high priority data.  

\subsection{Independent trials of local bias}

So far, our results have shown that utilizing Census has a much better cover time than local bias, through the use of short, temporary gradients for attracting the tokens towards unvisited nodes. However, there is one scenario in which just the use of local bias can achieve fast convergence without utilizing gradients --- the idea here is to use concurrent, independent trials of locally biased random walks under the assumption that the aggregation objective can tolerate redundant undetected copies in the result. We describe this optimization below.

Here we exploit the fact that until a significant fraction of nodes are visited (around $75\%$), locally biased random walks have very little wasted exploration (See Fig.~\ref{fig:2a}). In other words, the ratio of token transfers to that of visited nodes remains very small ($<2$) until this point and then rises steeply. So we introduce a constant number (about $4-6$, independent of network size) of concurrent random walk trials in the network with local biasing, and observe that by terminating individual trials before the inflection point, the union of these trials visits almost all the nodes in the network. 

\begin{figure*}[htbp]
\vspace*{3mm}
  \begin{center}
    \mbox{
      \subfigure[] {\scalebox{0.35}{\includegraphics[width=\textwidth]{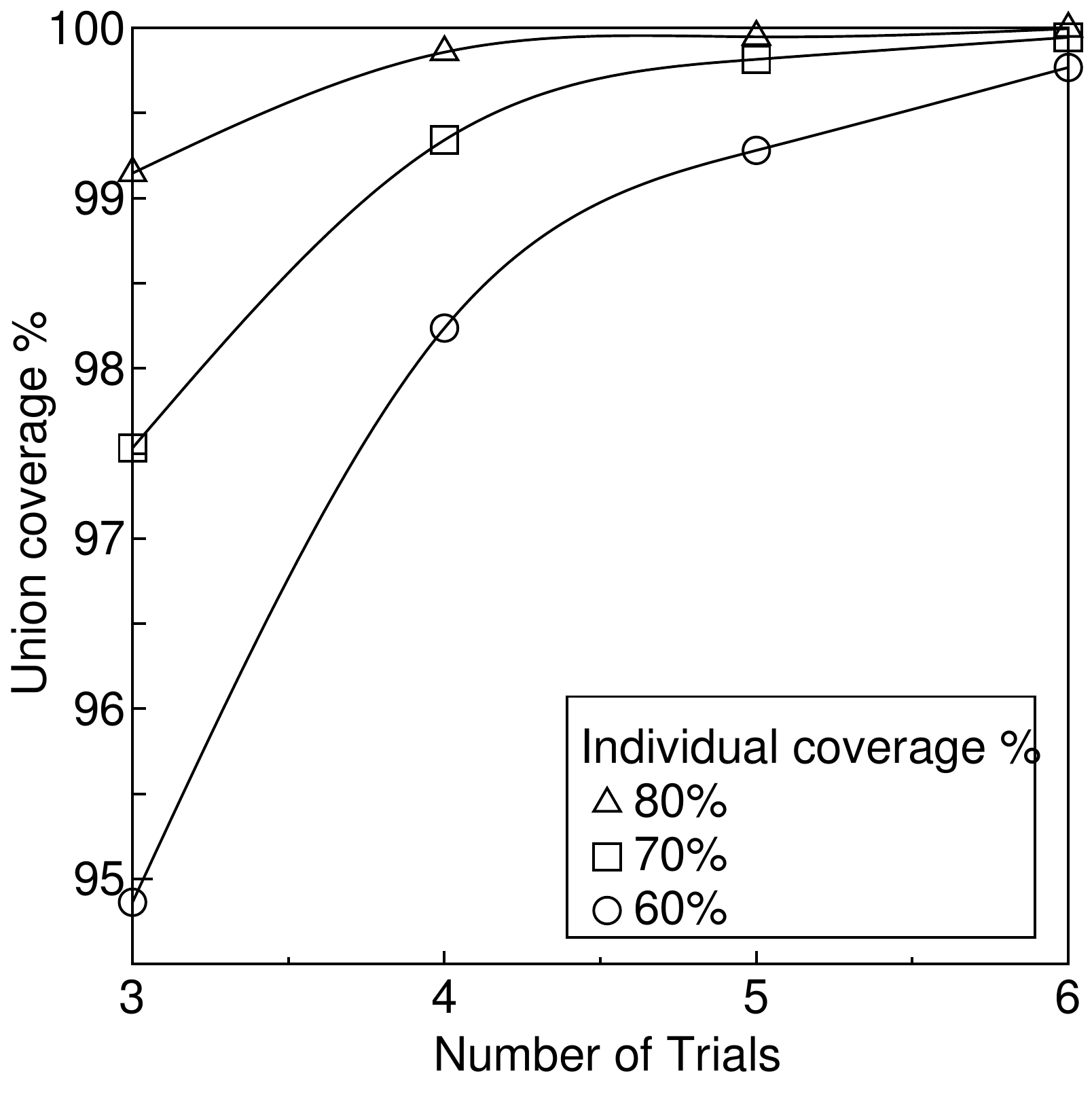}} \label{fig:8a}} \quad
      \subfigure[] {\scalebox{0.35}{\includegraphics[width=\textwidth]{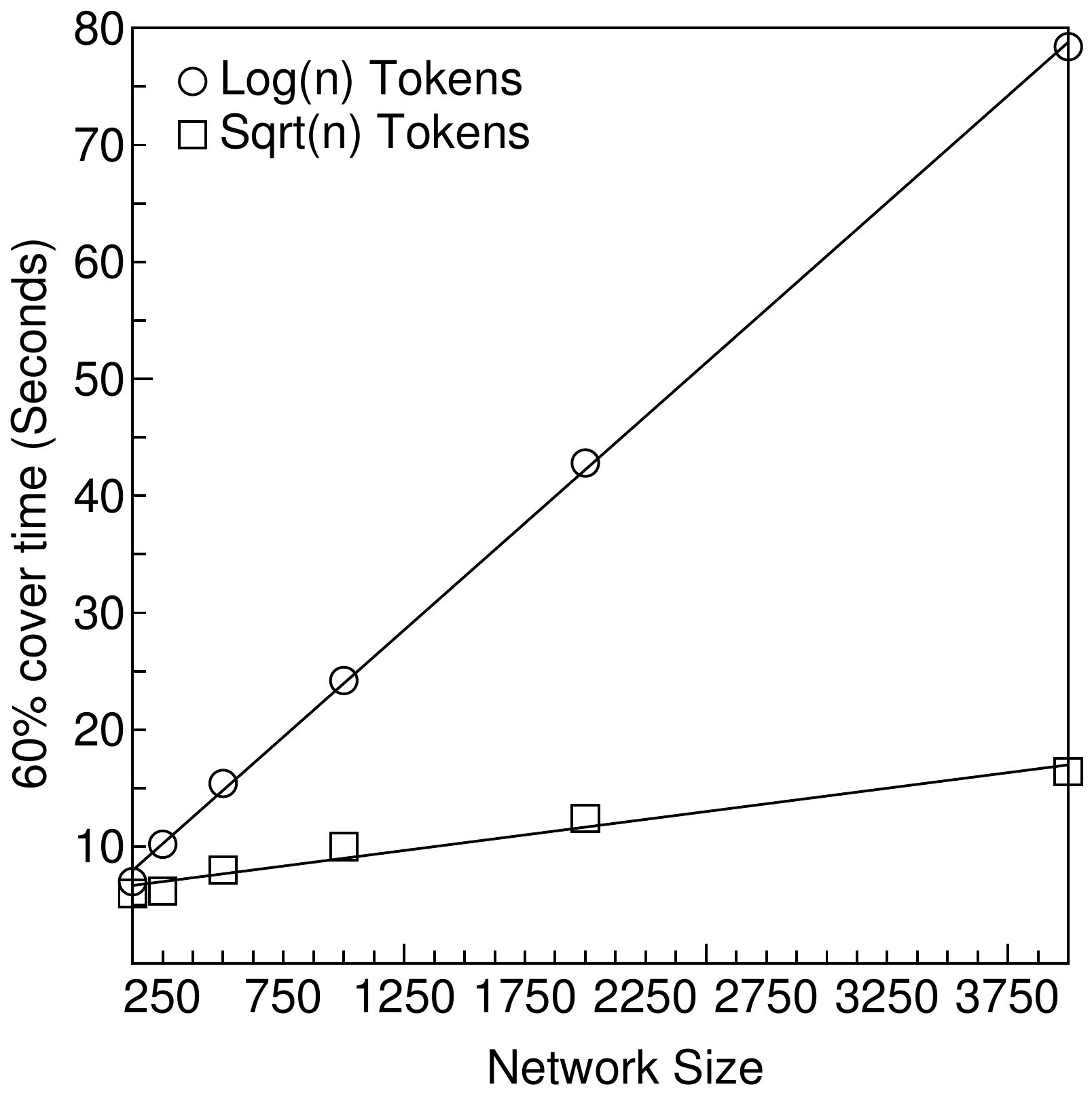}} \label{fig:8b}} \quad
      
      }
    \vspace{-2mm}  
    \caption{(a) Union coverage under multiple, independent trials at different levels of individual coverage (b) Cover time for $60\%$ coverage with local bias ($\sqrt{N}$ and $log(N)$ tokens)}
       \label{res8}
  \end{center}
\vspace*{-2mm}
\end{figure*}

In Fig.~\ref{fig:8a}, we show the combined coverage with different combinations of partial completion levels and number of trials. The numbers shown are the worst case numbers over all network sizes. As seen in the figure, $5$ trials each running to $60\%$ can cover $99\%$ of the nodes. In Fig.~\ref{fig:8b}, we show the $60\%$ cover time with local bias using $\sqrt{N}$ and $log(N)$ tokens and note that these times are much smaller than the corresponding $100\%$ cover times shown in Fig.~\ref{res3} for local bias. Also, recall from Fig.~\ref{fig:2a} that there is very little wasted exploration until a coverage of about $70\%$ -- the ratio of token passes to the number of visited nodes stay under $2$. Moreover, as seen in Fig.~\ref{fig:2a}, there is also very little variance across network sizes in terms of the ratio of token transfers to visited nodes -- {\em this fact can be leveraged to determine terminating conditions for individual trials}.

Thus, the cover time is shown to reduce significantly with this approach, while maintaining a very low exploration overhead. But, note that this optimization may not be appropriate for all types of aggregation (e.g., count and sum) because some node values may be reflected in multiple trial aggregates. The method is especially useful when the goal is to compute extremal values (such as max or min or range) of node state in the network \cite{learnhist}.

%if the network is connected, i.e., there are no partitions

\section{Related work on random walks}
Results on the cover time range for random walks in graphs have been shown to vary from $O(Nlog(N))$ for complete graphs to $O(N^3 )$ for lollipop graphs \cite{rw1, rw2, rw3}. Typically, cover times are lower for dense, highly connected graphs and tend to increase as connectivity decreases \cite{rw8}. A speed-up by a factor of $k$ has been shown when $k$ independent random walks are utilized in the graph \cite{rw4, rw5, rw6}. However, all these results have been obtained in the context of static graphs. Cover times for biased random walks in time-varying graphs (relevant for mobile networks), which are the focus of this paper, have not been studied to the best of our knowledge. 

That being said, our results on cover times for node visitation appear to be related to cover times of random walks on certain {\em static} geometric graphs \cite{rw7}, in which nodes are placed uniformly on a unit square and two nodes are connected if and only if their Euclidean distance is less than some radius $r$. For mesh networks modeled as geometric graphs with uniform degree of connectivity, the expected cover time is known to be $O(N log^2 N)$ \cite{rw7}. We have shown that this bound can be improved to $O(Nlog(N))$ with local biasing and $O(N)$ with Census.

We also note that in previous work there has been some empirical evidence of obtaining an $O(Nlog(N))$ cover time for {\em static 2d grids} by exploiting some form of choice in random walks, where preference is given to less visited nodes at each step \cite{rw8}. The variance in the number of visits per node for such processes has been studied in \cite{rw12}. The locally biased random walk explored in this paper is quite different than random walks with choice. In a locally biased random walk, unvisited nodes are preferred whenever they are available, and if there are no unvisited neighbors, the token is moved to any randomly selected neighbor. Nodes do not keep track of the number of times they have been visited, but rather just keep track of whether they have been visited. The motivation for local biasing is to simply move opportunistically towards unvisited nodes whenever they are within range. Moreover, while the results in \cite{rw8} are empirical, in this paper we have analytically shown the cover times for local and gradient biased random walks.

%that locally biasing a random walk results in a cover time of $O(Nlog(N))$ for random geometric graphs in a mobile setting. We also show that complementing the local bias with a multi-hop gradient bias improves the cover time to $O(N)$ and avoids a long tail.

% while maintaining the same order of convergence.	

\section{Conclusions}
In this paper, we introduced the idea of using biased random walks for the node visitation problem in MANETs. We observed that random walks per se have high cover times, and that using only local one-hop bias, where the token prefers an unvisited neighbor whenever available, reduces the coverage time significantly to an order complexity of $O(Nlog(N)/k)$ but exhibits a long tail before all nodes are covered. To redress these shortcomings, we introduced a temporary multi-hop gradient bias to pull the tokens towards unvisited nodes. The resulting protocol has a convergence time of $O(N/k)$, avoids a long tail and significantly reduces the coverage time as well as the token passing overhead. 

We quantified the performance of {\em Census} under different network conditions. Our analysis shows that gradient bias, outperforms unbiased as well as local bias random walks under all of these conditions.  We showed that {\em Census} has very little state overhead and is naturally resilient to topology changes of MANETs. We also showed how {\em Census} outperforms other techniques for node visitation such as flooding, gossip, and structure based routing protocols. Our analysis of {\em Census} is corroborated by simulations in ns-3 for networks ranging from 125-4000 nodes under different densities, mobility models, and node speeds in MANETs, thus demonstrating the scalability and robustness of {\em Census}. 

%Another significant advantage of {\em Census} is that it is extremely lightweight and makes minimal assumptions about the underlying network stack. For instance, {\em Census} does not require knowledge of node addresses, does not require a neighborhood discovery service, does not require knowledge of network topology, does not require a routing service (like IP) and does not assume localization. This makes it very easy to implement {\em Census} and port it directly across heterogeneous radio platforms. This also fits very well with a scalable Application Specific Networking Pattern (ASNP) architecture for MANETs that we are pursing in related work, in which multiple network protocols can co-exist over a minimalist link layer abstraction provided to them on top of one or more underlying radio platforms \cite{asnpref}. 

We note that the idea of successively visiting all nodes in the network using token passing has other applications beyond counting and data aggregation. For example, the protocol could also be used to provide every node an access to a critical resource in a mutually exclusive manner, such as the use of a shared high-bandwidth link. Moreover, in some operations such as averaging and computing histograms of network state, visiting almost all but not necessarily all nodes might well suffice.  For such operations, using only local one-hop bias with independent, concurrent trials might well be sufficient, as opposed to using multi-hop gradients.

Importantly, {\em Census} is lightweight in terms of resource requirements and makes rather minimal assumptions of the underlying network.  In particular, it does not assume knowledge of node addresses or locations, require a neighborhood discovery service or network topology information, or depend upon any particular routing or transport protocols such as TCP/IP.  A key implication is that {\em Census} is suitable for heterogeneous networks (and radio platforms).  Likewise, {\em Census} is not merely useful as an application layer protocol but can serve as a network routing protocol itself.  In this sense, we regard {\em Census} as a candidate for a MANET network architecture that allows for Application Specific Network Patterns (ASNPs) so as to achieve scalability and robustness.  These attributes have emerged as being important given recent high profile failures of MANETs, especially the JTRS program \cite{jtrs} which invested several billion dollars hoping to realize scalable, heterogeneous MANETs, that have led to a clean-slate redesign of MANET solutions \cite{fwid1, fwid2}.  The interested reader is referred to a companion document \cite{asnpref} on the ASNP architecture and its use of a minimal link layer abstraction to support multiple ASNPs and their applications on heterogeneous radios.

\singlespacing
\bibliographystyle{abbrv}
\bibliography{vinod}

\end{document}